\documentclass[12pt,a4paper]{article}
\usepackage{amssymb}
\usepackage{amsmath}
\usepackage[top=1in, bottom=1.25in, left=1.25in, right=1.25in]{geometry}
\usepackage[onehalfspacing]{setspace}
\usepackage{amsfonts}
\usepackage{amstext}
\usepackage{amsthm}
\usepackage[round]{natbib}
\usepackage{hyperref}
\usepackage{graphicx}
\usepackage{xr}

\numberwithin{equation}{section}
\newtheorem{theorem}{Theorem}

\newtheorem{condition}{Condition}

\newtheorem{definition}{Definition}

\newtheorem{lemma}{Lemma}

\newtheorem{proposition}{Proposition}
\theoremstyle{remark}
\newtheorem{remark}{Remark}

\input{tcilatex}

\begin{document}

\title{Best Subset Binary Prediction\thanks{{\footnotesize We are grateful
to the co-editor, Jianqing Fan, an associate editor and three anonymous
referees for constructive comments and suggestions. We also thank Stefan
Hoderlein, Joel Horowitz, Shakeeb Khan, Toru Kitagawa, Arthur Lewbel, and
participants in 2017 Asian Meeting of the Econometric Society and 2017
annual conference of the International Association for Applied Econometrics
for helpful comments. This work was supported by the Ministry of Science and
Technology, Taiwan (MOST106-2410-H-001-015-), Academia Sinica (Career
Development Award research grant), the European Research Council
(ERC-2014-CoG-646917-ROMIA), the UK Economic and Social Research Council
(ES/P008909/1 via CeMMAP), and the British Academy (International
Partnership and Mobility Scheme Grant, reference number PM140162).}}}
\author{Le-Yu Chen\thanks{%
E-mail: lychen@econ.sinica.edu.tw} \\
{\small {Institute of Economics, Academia Sinica}} \and Sokbae Lee\thanks{%
E-mail: sl3841@columbia.edu} \\
{\small {Department of Economics, Columbia University}}\\
{\small {Centre for Microdata Methods and Practice, Institute for Fiscal
Studies} }}
\date{May 2018}
\maketitle

\begin{abstract}
\noindent We consider a variable selection problem for the prediction of
binary outcomes. We study the best subset selection procedure by which the
covariates are chosen by maximizing \citet{Manski:1975,
Manski:1985}'s maximum score objective function subject to a constraint on
the maximal number of selected variables. We show that this procedure can be
equivalently reformulated as solving a mixed integer optimization problem,
which enables computation of the exact or an approximate solution with a
definite approximation error bound. In terms of theoretical results, we
obtain non-asymptotic upper and lower risk bounds when the dimension of
potential covariates is possibly much larger than the sample size. Our upper
and lower risk bounds are minimax rate-optimal when the maximal number of
selected variables is fixed and does not increase with the sample size. We
illustrate usefulness of the best subset binary prediction approach via
Monte Carlo simulations and an empirical application of the work-trip
transportation mode choice.




\end{abstract}

\section{Introduction\label{model}}

Prediction of binary outcomes is an important topic in economics and various
scientific fields. Let $Y\in \{0,1\}$ be the binary outcome of interest and $%
W$ a vector of covariates for predicting $Y$. Assume that the researcher has
a training sample of $n$ independent identically distributed (i.i.d.)
observations $\left( Y_{i},W_{i}\right) _{i=1}^{n}$ of $(Y,W)$. For $w\in 
\mathcal{W}$, let 
\begin{equation}
b_{\theta }(w)\equiv 1\left\{ w^{\prime }\theta \geq 0\right\} ,
\label{predictor}
\end{equation}%
where $\mathcal{W}$ is the support of $W$, $\theta $ is a vector of
parameters, and $1\left\{ \cdot \right\} $ is an indicator function that
takes value 1 if its argument is true and 0 otherwise.

One reasonable prediction rule is to choose $\theta $ such that it maximizes
the probability of making the correct prediction $P\left[ Y=b_{\theta }(W)%
\right] $. However, this is infeasible in practice since the joint
distribution of $(Y,W)$ is unknown. A natural sample analog is to maximize
the sample average score which equals the proportion of correct predictions
under the prediction rule (\ref{predictor}) in the training sample. This
maximization problem is equivalent to the maximum score estimation in binary
response models and is pioneered by \citet{Manski:1975, Manski:1985}. Thus,
we call the corresponding prediction rule the \emph{maximum score }%
prediction rule. See \citet{Manski:Thompson:1989}, \citet{Jiang:10}, and 
\citet{
Elliott2013} for prediction in the maximum score approach.

This prediction problem has the same structure as the binary classification
problem, which is extensively studied in the statistics and machine learning
literature. For example, see the classic work of \citet{DGL:1996} among many
others. In this literature, the empirical risk minimization (ERM) classifier
over the class of binary classifiers specified by (\ref{predictor}) is
defined as a minimizer of the empirical predictive risk, which is taken to
be one minus the objective function of the maximum score prediction problem.
In other words, the ERM classification rule is identical to the maximum
score prediction rule.

In this paper, we address the covariate selection issue in the framework of
predicting the binary outcome $Y_{i}$ using the class of linear
threshold-crossing prediction rules $b_{\theta }(W_{i})$ defined by (\ref%
{predictor}). We study the \emph{best subset }selection procedure by which
the covariates are chosen among a collection of candidate explanatory
variables by maximizing the empirical score subject to a constraint on the
maximal number of selected variables. In other words, we investigate
theoretical and numerical properties of the $\ell _{0}$-norm constrained
maximum score prediction rules.\footnote{%
Here, the $\ell _{0}$-norm of a real vector refers to the number of non-zero
components of the vector.}

To the best of our knowledge, \citet{Greenshtein2006} and \citet{Jiang:10}
are the only existing papers in the literature that explicitly considered
the same prediction problem as ours. \citet{Greenshtein2006} considered a
general loss function that includes maximum score prediction as a special
case in the i.i.d.\thinspace setup. \citet{Greenshtein2006} focused on a
high dimensional case and established conditions under which the excess risk
converges to zero as $n\rightarrow \infty $. \citet{Jiang:10} focused on the
prediction of time series data and obtained an upper bound for the excess
risk. Neither \citet{Greenshtein2006} nor \citet{Jiang:10} provided any
numerical results for the best subset maximum score prediction rule. In
contrast, we focus on cross-sectional applications and emphasize
computational aspects.

The main contributions of this paper are twofold: first, we show that the
best subset maximum score prediction rule is minimax rate-optimal and
second, we demonstrate that it can be implemented via mixed integer
optimization. The first contribution is theoretical and builds on the
literature of empirical risk minimization 
\citep[][in
particular]{tsybakov2004, massart2006}. Specifically, we obtain \emph{%
non-asymptotic} upper and lower risk bounds when the dimension of potential
covariates is possibly much larger than the sample size $n$. Our upper and
lower risk bounds are minimax rate-optimal when the maximal number of
selected variables is fixed and does not increase with $n$. The existing
results of finite-sample upper and lower risk bounds for the binary
prediction problem focus on the case where there is no variable selection
and the set of covariates is fixed and low-dimensional. Our risk bound
results extend to the setup under the $\ell _{0}$-norm constraint when the
set of potential covariates is high-dimensional.\footnote{%
\citet{Raskutti2011} developed minimax rate results for high-dimensional
linear mean regression models. We have used in the derivation of our lower
risk bound a technical lemma of their paper \citep[][Lemma
4]{Raskutti2011}, which is based on the approximation theory literature.
Nonetheless, our results are not directly obtainable from %
\citet{Raskutti2011}, who considered the least squares objective function.}

The second contribution is computational. We face two kinds of computational
challenges. One challenge comes from the nature of the objective function
and the other is from the best subset selection. The maximum score objective
function is a piecewise constant function whose range set contains only
finitely many points. Hence the maximum of the score maximization problem is
always attained yet the maximizer is generally not unique. It is known that
computing the maximum score estimates regardless of the presence of the $%
\ell _{0}$ constraint is NP (non-deterministic polynomial-time)-hard 
\citep[see,
e.g.,][]{Johnson:78}. See \citet{Manski:Thompson:1986} and \citet{PINKSE1993}
for first generation algorithms for maximum score estimation.

Our computation algorithm is based on the method of mixed integer
optimization (MIO). \citet{Florios:Skouras:08} provided compelling numerical
evidence that the MIO approach is superior to the first-generation
approaches. \citet{Kitagawa:Tetenov:2015} used an MIO formulation that is
different from \citet{Florios:Skouras:08} to solve maximum score type
problems. The objective of interest in \citet{Kitagawa:Tetenov:2015} is to
develop treatment choice rules by maximizing an empirical welfare criterion,
which resembles the maximum score objective function. They derived minimax
optimality and used the MIO formulation to implement their algorithm.
Neither \citet{Florios:Skouras:08} nor \citet{Kitagawa:Tetenov:2015} were
concerned with the variable selection problem.

These second generation approaches are driven by developments in MIO solvers
and also by availability of a much faster computer compared to the period
when the first generation algorithms were proposed. %
\citet{Florios:Skouras:08} reported that they obtained the exact maximum
score estimates using \citet{horowitz1993}'s data in 10.5 hours. In this
application, the sample size was $n=842$ and there were 4 parameters to
estimate.

It is well known that use of a good and tighter parameter space can
strengthen the performance of a global optimization procedure including the
MIO approach. In this paper, we propose a data driven approach to refine the
parameter space. Using a state-of-the-art MIO solver as well as a
tailor-made heuristic to choose the parameter space, it took us less than 5
minutes to obtain the exact maximum score estimates using the same dataset
with the same number of parameters to estimate.\footnote{%
This numerical result can be found in Online Appendix \ref{emp:linear:spec}
of the paper.} This is a dramatic improvement at the factor of more than 100
relative to the numerical performance reported in \citet{Florios:Skouras:08}%
. In other words, we demonstrate that hardware improvements combined with
the advances in MIO solvers and also with a carefully chosen parameter space
have made the maximum score approach empirically much more relevant now than
ten years ago.

The second numerical challenge is concerned with constrained optimization
with the $\ell _{0}$-norm constraint. It is well known that the $\ell _{0}$%
-norm constraint renders the variable selection problem NP-hard even in the
regression setup where the objective function is convex and smooth 
\citep[see,
e.g.,][]{Natarajan:95,bertsimas2016}. Recently \citet{bertsimas2016}
proposed a novel MIO approach to the best subset variable selection problem
when least squares and least absolute deviation risks are concerned. They
demonstrated that the MIO approach can efficiently deliver a provably
optimal solution to the resulting $\ell _{0}$-norm constrained risk
minimization problem for a variety of datasets with practical problem size.
Our implementation of the best subset maximum score prediction rules
combines insights from \citet{bertsimas2016}, \citet{Florios:Skouras:08},
and \citet{Kitagawa:Tetenov:2015}. We present two alternative MIO solution
methods that complement each other.

In practical applications, it is useful to consider an approximate solution
by adopting an early termination rule. In our empirical application, by
setting an explicitly pre-specified optimization error, we were able to
obtain approximate maximum score estimates with \citet{horowitz1993}'s data
in around 10 minutes when both an intercept term and one specific random
covariate were always selected, and there were 9 additional auxiliary
covariates that were subject to the constraint where at most 5 of them could
be selected. This suggests that fast developments in computing environments
will enable us to solve an empirical problem at a practically relevant scale
in very near future. We provide additional numerical evidence in Monte Carlo
experiments in a high-dimensional setup when the number of potential
covariates is larger than the sample size.

The remainder of this paper is organized as follows. In Section \ref%
{covariate selection}, we describe our prediction rule. Section \ref{risk
performance} establishes theoretical properties of the proposed prediction
rule. In Section \ref{implementation}, we present computation algorithms
using the MIO approach, and in Section \ref{simulation study}, we conduct a
simulation study on the performance of our prediction rule in both low and
high dimensional variable selection problems. In Section \ref{empirical
application}, we illustrate usefulness of our prediction rule in the
empirical application of work-trip mode choice using \citet{horowitz1993}'s
data. We then conclude the paper in Section \ref{conclusion}. Proofs of all
theoretical results and supplementary material of this paper are collated in
online appendices.

\section{A Best Subset Approach to Maximum Score Prediction of Binary
Outcomes}

\label{covariate selection}

In this section, we describe our proposal of the best subset maximum score
prediction rule. Following \citet{MD:99} and \citet{Danilov2004}, we
distinguish between \emph{focus covariates} that are always included in the
prediction rule and \emph{auxiliary covariates} of which we are less
certain. We thus decompose the covariate vector $W$ as $W=(X,Z)$, where $X$
is a $\left( k+1\right) $-dimensional vector of focus covariates and $Z$ is
a $p$-dimensional vector of auxiliary covariates.

Noting that $b_{a\theta }(w)=$ $b_{\theta }(w)$ for any positive real scalar 
$a$, we adopt the same scale normalization method as in \citet{Horowitz:92}
and \citet{Jiang:10} by restricting the magnitude of the coefficient of one
of the focus covariates to be unity. Specifically, write $X=(X_{0},%
\widetilde{X})$ where $X_{0}$ is a scalar variable and $\widetilde{X}$ is
the remaining $k$-dimensional subvector of focus covariates. The parameter
vector $\theta $ in (\ref{predictor}) is decomposed accordingly as $\theta
=(\alpha ,\beta ,\gamma )$, where $\alpha \in \{-1,1\}$ and $\left( \beta
,\gamma \right) \in \Theta $, which is a subset of $\mathbb{R}^{k+p}$. In
this notation, the binary prediction rule has the following form: 
\begin{equation}
b_{\alpha ,\beta ,\gamma }(w)=1\{\alpha x_{0}+\widetilde{x}^{\prime }\beta
+z^{\prime }\gamma \geq 0\}\text{ for }w\in \mathcal{W}.
\label{augmented model}
\end{equation}

We consider a parsimonious variable selection method by which the
constituted prediction rule does not include more than a pre-specified
number of auxiliary covariates. For any $p$ dimensional real vector $c$, let 
$\left\Vert c\right\Vert _{0}\equiv\sum\nolimits_{j=1}^{p}1\{c_{j}\neq0\}$
be the $\ell_{0}$-norm of $c$. We carry out the $\ell_{0}$-norm constrained
covariate selection procedure by solving the constrained maximization problem%
\begin{equation}
\max\nolimits_{\left( \alpha,\beta,\gamma\right) \in\{-1,1\}\times\Theta
_{q}}\text{ }S_{n}(\alpha,\beta,\gamma ),
\label{best subset maximum score binary prediction rule}
\end{equation}
where the objective function $S_{n}$ is defined as 
\begin{equation}
S_{n}(\theta )\equiv n^{-1}\sum\nolimits_{i=1}^{n}1\{Y_{i}=b_{\theta
}(W_{i})\}  \label{empirical objective function}
\end{equation}%
and the $\ell_0$-norm constrained parameter space is given as 
\begin{equation}
\Theta_{q}\equiv\{\left( \beta,\gamma\right) \in\Theta\subset\mathbb{R}%
^{k+p}:\left\Vert \gamma\right\Vert _{0}\leq q\}  \label{theta_q}
\end{equation}
for a given positive integer $q$.

As we discussed in the introduction, solving for the exact maximizer for (%
\ref{best subset maximum score binary prediction rule}) is desirable yet can
be computationally challenging. It is hence practically useful to consider
an approximate solution, which is constructed below, to the maximization
problem (\ref{best subset maximum score binary prediction rule}).

For any $\varepsilon \geq 0$, let $(\widehat{\alpha },\widehat{\beta },%
\widehat{\gamma })\in \{-1,1\}\times \Theta _{q}$ be an approximate
maximizer with $\varepsilon $ tolerance level such that%
\begin{equation}
S_{n}(\widehat{\alpha },\widehat{\beta },\widehat{\gamma })\geq
\max\nolimits_{\left( \alpha ,\beta ,\gamma \right) \in \{-1,1\}\times
\Theta _{q}}\text{ }S_{n}(\alpha ,\beta ,\gamma )-\varepsilon \text{ almost
surely.}  \label{epsilon level best subset MSBP rule}
\end{equation}%
We refer to the prediction rule defined by $1\{\widehat{\alpha }x_{0}+%
\widetilde{x}^{\prime }\widehat{\beta }+z^{\prime }\widehat{\gamma }\geq 0\}$
as the approximate best subset maximum score binary prediction rule.%
\footnote{%
The dependence of $(\widehat{\alpha },\widehat{\beta },\widehat{\gamma })$
on $\varepsilon $ is suppressed for simplicity of notation.} The value of $%
\varepsilon $ can be specified for early termination of the solution
algorithm to the problem (\ref{best subset maximum score binary prediction
rule}). In Section \ref{implementation}, we will present an algorithm that
allows for computing an approximate solution to (\ref{best subset maximum
score binary prediction rule}) within a definite approximation error bound
specified by the tolerance level $\varepsilon $. In what follows, we use
PRESCIENCE as shorthand for the approximate best subset maximum score binary
prediction rule.\footnote{%
It comes from the aPpRoximate bEst S(C)ubset maxImum scorE biNary prediCtion
rulE.}

\begin{remark}
In terms of model selection, there are several aspects one needs to
consider. First, one needs to specify the covariate vector $W$. We recommend
starting with a large set of covariates for $W$ since we have a built-in
model selection procedure. Second, it is necessary to decide which
covariates belong to $X$ (focus covariates) and $Z$ (auxiliary covariates).
What consists of auxiliary covariates depends on particular applications.
The auxiliary covariates correspond to the part of the model specification
the researcher is not sure about. For example, they could be some higher
order terms or interaction terms. If a researcher does not have concrete
ideas about the auxiliary covariates, we recommend letting the auxiliary
covariates be all regressors except one the researcher is specifically
interested in. Third, it is required to choose $q$ (the $\ell _{0}$-norm
constraint). The constant $q$ is an important tuning parameter in our
procedure. A particular choice of $q$ can be motivated in some specific
applications. Generally speaking, for the purpose of prediction, there is
the standard tradeoff between flexibility, which requires a larger $q$, and
the risk of over-fitting, which pushes for a smaller $q$. We recommend using
cross validation to choose $q$, as we will demonstrate in our empirical
example and Monte Carlo experiments.
\end{remark}

\section{Theoretical Properties of PRESCIENCE}

\label{risk performance}

In this section, we study the theoretical properties of PRESCIENCE. Let $F$
denote the joint distribution of $\left( Y,W\right) $. For $\left( \alpha
,\beta ,\gamma \right) \in \{-1,1\}\times \Theta $, let 
\begin{equation}
S(\alpha ,\beta ,\gamma )\equiv P\left( Y=b_{\alpha ,\beta ,\gamma
}(W)\right) .  \label{S(a,b,c)}
\end{equation}%
Note that $S(\alpha ,\beta ,\gamma )$ depends on the joint distribution $F$.
Given a cardinality bound $q$, let 
\begin{equation}
S_{q}^{\ast }\equiv \sup\nolimits_{\left( \alpha ,\beta ,\gamma \right) \in
\{-1,1\}\times \Theta _{q}}\text{ }S(\alpha ,\beta ,\gamma ).
\label{theoretical prediction rule}
\end{equation}%
That is, $S_{q}^{\ast }$ is the supremum of $S(\alpha ,\beta ,\gamma )$
given the $\ell _{0}$-norm constraint.

Following the literature on empirical risk minimization (see, e.g., %
\citet{DGL:1996}, \citet{Lugosi:2002}, \citet{tsybakov2004}, %
\citet{massart2006}, \citet{Greenshtein2006} and \citet{Jiang:10} among many
others), we assess the predictive performance of PRESCIENCE by bounding the
difference 
\begin{equation}
U_{n}\equiv S_{q}^{\ast }-S(\widehat{\alpha },\widehat{\beta },\widehat{%
\gamma }),  \label{U}
\end{equation}%
where $(\widehat{\alpha },\widehat{\beta },\widehat{\gamma })$ is defined by
(\ref{epsilon level best subset MSBP rule}). The difference $U_{n}$ is
non-negative by the definition of $S_{q}^{\ast }$. Hence, a good prediction
rule will result in a small value of $U_{n}$ with a high probability and
also on average.

Throughout this section, we assume that $\varepsilon =0$ for simplicity.
Before presenting our theoretical results, we first introduce some notation.
For any two real numbers $a$ and $b$, let $a\vee b\equiv \max \{a,b\}$. Let $%
s\equiv k+q$ and 
\begin{equation}
r_{n}\equiv s\ln (p\vee n)\vee 1.  \label{rn}
\end{equation}

\begin{theorem}
\label{maximal inequality} Assume $s\geq1$. Then for all $\sigma>0$, there
is a universal constant $M_{\sigma}$, which depends only on $\sigma$, such
that 
\begin{equation}
P\left( U_{n}>2\sqrt{\frac{M_{\sigma}r_{n}}{n}}\right) \leq \exp(-\sigma
r_{n}),  \label{bound}
\end{equation}
provided that 
\begin{equation}
\left( 4s+4\right) \ln\left( M_{\sigma}r_{n}\right) \leq r_{n}+\left(
6s+5\right) \ln2.  \label{condition on r_n}
\end{equation}
\end{theorem}

Theorem \ref{maximal inequality} establishes that the tail probability of $%
U_{n}$ decays exponentially in $r_{n}$. Moreover, this result is
non-asymptotic: inequality (\ref{bound}) is valid for every sample size $n$
for which condition (\ref{condition on r_n}) holds. By comparing the leading
terms on both sides of inequality (\ref{condition on r_n}), we can see that
condition (\ref{condition on r_n}) is satisfied, for instance, if 
\begin{equation}
4 [\ln(s) + \ln (\ln (p \vee n)) + \ln(M_{\sigma})] \leq \frac{1}{2} \ln (p
\vee n).  \label{suff condition on r_n}
\end{equation}
Hence, condition (\ref{condition on r_n}) is satisfied easily when $p\vee n$
takes a relatively large value compared to $s$. If $(p\vee n)$ diverges to
infinity, then $s$ can diverge at a sufficiently slow rate.

Theorem \ref{maximal inequality} implies that 
\begin{equation}
E\left[ U_{n}\right] =O\left( n^{-1/2}\sqrt{s\ln (p\vee n)}\right) =o(1),
\label{theorem1-expectation}
\end{equation}%
provided that 
\begin{equation}
s\ln (p\vee n)=o(n)  \label{rate}
\end{equation}%
holds. This allows the case that 
\begin{equation}
\ln p=O(n^{\alpha })\text{ and }s=o(n^{1-\alpha })\text{ for }0<\alpha <1.
\label{rate result}
\end{equation}%
In other words, the predictive performance of PRESCIENCE remains good even
when the number of potentially relevant covariates ($p$) grows
exponentially, provided that the number of selected covariates ($s$) can
only grow at a polynomial rate. \citet{GreenshteinRitov2004} and %
\citet{Greenshtein2006} consider the case where $p$ grows at a polynomial
rate. For this case, condition (\ref{rate}) implies that $s=o\left( n/\ln
n\right) $, which coincides with the optimal sparsity rate established by %
\citet{GreenshteinRitov2004} and \citet{Greenshtein2006} under which a
sequence of predictor selection procedures subject to the sparsity
constraint can be shown to be persistent.

\begin{remark}
For the case with $\varepsilon >0$, it is straightforward to modify the
theoretical results presented above such that the rate result (\ref%
{theorem1-expectation}) continues to hold provided that 
\begin{equation}
\varepsilon =O\left( n^{-1/2}\sqrt{s\ln (p\vee n)}\right) .  \label{epsilon}
\end{equation}
\end{remark}

\subsection{An Upper Bound under the Margin Condition}

The result (\ref{theorem1-expectation}) is derived under the
i.i.d.\thinspace setup but does not hinge on other regularity conditions on
the underlying data generating distribution $F$. This rate result can be
improved under additional assumptions on the distribution $F$. In this
section, we consider a condition that is called the \emph{margin condition}
in the literature under which we may obtain a sharper result on the upper
bound of $E\left[ U_{n}\right] $. As before, the derived bound will be
non-asymptotic.

It is necessary to introduce additional notation. Let%
\begin{equation}
\mathcal{B}_{q}\mathcal{\equiv }\left\{ b_{\theta }:\theta \in
\{-1,1\}\times \Theta _{q}\right\} .  \label{Bq}
\end{equation}%
That is, $\mathcal{B}_{q}$ is the class of all prediction rules in %
\eqref{augmented model} with the $\ell _{0}$-norm constraint. For $w\in 
\mathcal{W}$, let%
\begin{align}
\eta (w)& \equiv P(Y=1|W=w),  \label{eta_w} \\
b^{\ast }(w)& \equiv 1\left\{ \eta (w)\geq 0.5\right\} .  \label{b_star}
\end{align}%
For any measurable function $f:\mathcal{W\mapsto }\mathbb{R}$, let $%
\left\Vert f\right\Vert _{1}=E\left[ \left\vert f(W)\right\vert \right] $
denote the $L_{1}$-norm of $f$. The functions $\eta $ and $b^{\ast }$ as
well as the $L_{1}$-norm $\left\Vert \cdot \right\Vert _{1}$ depend on the
data generating distribution $F$. For any indicator function $b:\mathcal{W}%
\mapsto \left\{ 0,1\right\} $, let 
\begin{equation}
\widetilde{S}\left( b\right) \equiv P\left( Y=b(W)\right) .  \label{S_tilda}
\end{equation}%
We now state the following regularity condition.

\begin{condition}[Margin Condition]
\label{margin condition for upper bound}There are some $\vartheta \geq 1$
and $h>0$ such that, for every binary predictor $b:\mathcal{W}\mapsto
\{0,1\} $, 
\begin{equation}
\widetilde{S}\left( b^{\ast }\right) -\widetilde{S}\left( b\right) \geq
h^{\vartheta }\left\Vert b^{\ast }-b\right\Vert _{1}^{\vartheta }.
\label{well posedness}
\end{equation}
\end{condition}

Condition \ref{margin condition for upper bound} is termed as the margin
condition in the literature (see, e.g., \citet{mammen1999}, %
\citet{tsybakov2004} and \citet{massart2006}). For any binary predictor $b$, 
\begin{equation}
\widetilde{S}(b^{\ast })-\widetilde{S}(b)=E\left[ \left\vert 2\eta
(W)-1\right\vert \left\vert b^{\ast }(W)-b(W)\right\vert \right] ,
\label{objective value difference}
\end{equation}%
so that $\widetilde{S}(b)$ is maximized at $b=b^{\ast }$. Hence, Condition %
\ref{margin condition for upper bound} implies that the functional $%
\widetilde{S}\left( \cdot \right) $ has a well-separated maximum. Suppose
that there exist universal positive constants $C$ and $\alpha $ such that 
\begin{equation*}
P\left( |\eta (W)-1/2|\leq t\right) \leq Ct^{\alpha }
\end{equation*}%
for all $t>0$. Then by modifying the proof of Proposition 1 of %
\citet{tsybakov2004} slightly, we can show that \eqref{well posedness} holds
with $\vartheta =(1+\alpha )/\alpha $. See \citet{tsybakov2004} for further
discussions on the margin condition.

Recall that it is not necessary to assume \eqref{well posedness} to
establish the risk consistency, as shown in Theorem \ref{maximal inequality}%
. We show below that we can obtain a tighter upper bound on $E\left[ U_{n}%
\right] $ under \eqref{well posedness}. Let%
\begin{equation}
\rho _{n}\equiv 1\vee \left[ \ln 2+q\ln p+\left( s+1\right) \ln \left(
n+1\right) \right] .  \label{tau_n}
\end{equation}%
The next theorem, which is an application of \citet[Theorem 2]{massart2006},
establishes a finite-sample bound on $E\left[ U_{n}\right] $ under the
margin condition.

\begin{theorem}
\label{risk upper bound}There are universal constants $K$ and $K^{\prime }$
such that%
\begin{equation}
E\left[ U_{n}\right] \leq \left[ \widetilde{S}\left( b^{\ast }\right)
-\sup\nolimits_{b\in \mathcal{B}_{q}}\widetilde{S}\left( b\right) \right]
+K^{\prime }\left( \frac{K^{2}\rho _{n}}{nh}\right) ^{\vartheta /\left(
2\vartheta -1\right) },  \label{finite sample upper bound}
\end{equation}%
provided that Condition \ref{margin condition for upper bound} holds with 
\begin{equation}
h\geq \left( \frac{K^{2}\rho _{n}}{n}\right) ^{\frac{1}{2\vartheta }}.
\label{condition on h}
\end{equation}
\end{theorem}

For $p\vee n$ sufficiently large, we have that $\rho _{n}\leq 5s\ln (p\vee
n) $; thus, inequality (\ref{condition on h}) can hold under condition (\ref%
{rate}) in large samples, provided that $h$ is fixed or does not go to zero
too rapidly.

The first term on the right-hand side of inequality 
\eqref{finite sample
upper bound} represents the bias term. Equation (\ref{objective value
difference}) implies that there is no bias term, namely $\widetilde{S}\left(
b^{\ast }\right) =\sup\nolimits_{b\in \mathcal{B}_{q}}\widetilde{S}\left(
b\right) $ if $b^{\ast }\in \mathcal{B}_{q}$. Therefore, Theorem \ref{risk
upper bound} implies that%
\begin{equation}
E\left[ U_{n}\right] =O\left( \left[ \frac{s\ln (p\vee n)}{nh}\right]
^{\vartheta /\left( 2\vartheta -1\right) }\right) ,  \label{sharper rate}
\end{equation}%
provided that $b^{\ast }\in \mathcal{B}_{q}$.\footnote{%
For the case with $\varepsilon >0$, it is also straightforward to modify
Theorem \ref{risk upper bound} such that the rate result (\ref{sharper rate}%
) continues to hold provided that 
\begin{equation*}
\varepsilon =O\left( \left[ \frac{s\ln (p\vee n)}{nh}\right] ^{\vartheta
/\left( 2\vartheta -1\right) }\right) .
\end{equation*}%
} The rate of convergence in \eqref{sharper rate} doubles that in %
\eqref{theorem1-expectation} when $h$ is fixed and $\vartheta =1$. We notice
that, if $b^{\ast }\notin \mathcal{B}_{q}$, the upper bound derived in
Theorem \ref{risk upper bound} would asymptotically reduce to the non-zero
bias term $\widetilde{S}\left( b^{\ast }\right) -\sup\nolimits_{b\in 
\mathcal{B}_{q}}\widetilde{S}\left( b\right) $ and hence the margin
condition alone does not suffice for deducing there is improved rate of
convergence. Nevertheless, the rate result \eqref{theorem1-expectation}
still holds regardless of the validity of the presumption that $b^{\ast }\in 
\mathcal{B}_{q}$.

We now remark on the condition that $b^{\ast }\in \mathcal{B}_{q}$ in the
context of the binary response model specified below. Suppose that the
outcome $Y$ is generated from a latent variable threshold crossing model 
\citep[see,
e.g.,][]{Manski:1975, Manski:1985}:%
\begin{equation}
Y=1\{W^{\prime }\theta ^{\ast }\geq \xi \},  \label{binary response model}
\end{equation}%
where $\theta ^{\ast }$ denotes the true data generating parameter vector
and $\xi $ is an unobserved latent variable whose distribution satisfies that%
\begin{equation}
Med(\xi |W=w)=0\text{ for }w\in \mathcal{W}.  \label{median independence}
\end{equation}%
Let $\theta _{0}\equiv \arg \sup\nolimits_{\theta \in \{-1,1\}\times \Theta
_{q}}$ $S(\theta )$. For simplicity, assume that $\theta _{0}\in
\{-1,1\}\times \Theta _{q}$ so that the maximum is attained.

\begin{proposition}
\label{observational equivalence} Assume that the model given by (\ref%
{binary response model}) and (\ref{median independence}) is correctly
specified. Suppose that Condition \ref{margin condition for upper bound}
holds. Then $b^{\ast }\in \mathcal{B}_{q}$ if and only if $W^{\prime }\theta
_{0}$ and $W^{\prime }\theta ^{\ast }$ have the same sign with probability $%
1 $.
\end{proposition}

\citet[Proposition 2]{Manski:1988} showed that, for the binary response
model specified by (\ref{binary response model}) and (\ref{median
independence}), the true parameter value $\theta ^{\ast }$ is identified
relative to another value $\theta $ if and only if the event that $W^{\prime
}\theta $ and $W^{\prime }\theta ^{\ast }$ have different sign occurs with
positive probability. Therefore, Proposition \ref{observational equivalence}
implies that $b^{\ast }\in \mathcal{B}_{q}$ if and only if the
``pseudo-true'' value $\theta _{0}$ is observationally equivalent to $\theta
^{\ast }$. In particular, this implies that $\theta ^{\ast }=\theta _{0}$ if 
$\theta ^{\ast }$ is point-identified.

It would be interesting to study the role of the bias when $b^{\ast }\notin 
\mathcal{B}_{q}$ using the framework of sieve estimation \citep{Chen:2007}.
As pointed by \citet[Proposition 1]{Elliott2013}, what matters is how well
we can approximate the value of optimum $\sup\nolimits_{\theta \in
\{-1,1\}\times \Theta _{q}}S(\theta )$, not the optimizer $\arg
\sup\nolimits_{\theta \in \{-1,1\}\times \Theta _{q}}S(\theta )$. However,
it would be much more demanding to develop non-asymptotic theory when the
bias is present in our framework. We leave this as a topic for future
research.

\subsection{A Minimax Lower Bound under the Margin Condition}

In this section, we derive a minimax lower bound under the margin condition.
In particular, we focus on the case that $s=k+q$ is low-dimensional in that $%
s$ does not grow with sample size $n$ and also consider a sufficient
condition for the margin condition.

Condition \ref{margin condition for upper bound} is satisfied with $%
\vartheta =1$ whenever 
\begin{equation}
\left\vert 2\eta (w)-1\right\vert \geq h\text{ for }w\in \mathcal{W}.
\label{margin condition}
\end{equation}%
\citet{massart2006} introduced (\ref{margin condition}) as an easily
interpretable margin condition requiring that the conditional probability $%
\eta (w)$ should be bounded away from $1/2$. Condition (\ref{margin
condition}) holds under certain regularity assumptions on the binary
response model as indicated in the following proposition.

\begin{proposition}
\label{margin condition for the binary response model} Assume that the model
given by (\ref{binary response model}) and (\ref{median independence}) is
correctly specified. Suppose that there are universal constants $\kappa
_{1}>0,$ $\kappa _{2}>0$ such that (i) $P\left( \left\vert W^{\prime }\theta
^{\ast }\right\vert \geq \kappa _{1}\right) =1$ and (ii) there is some open
interval $T$ containing $\left( -\kappa _{1},\kappa _{1}\right) $ such that $%
P\left( \xi \leq t|W=w\right) $ has a derivative (with respect to $t$) which
is bounded below by $\kappa _{2}$ for every $t\in T$. Then condition (\ref%
{margin condition}) holds with $h=2\kappa _{1}\kappa _{2}$.
\end{proposition}

Conditions (i) and (ii) in Proposition \ref{margin condition for the binary
response model} assume that $\left\vert W^{\prime }\theta ^{\ast
}\right\vert $ is bounded away from zero and the density of $\xi $
conditional on $W=w$ is bounded away from zero in a neighborhood of zero.
While the latter condition is mild, the former is non-trivial. Condition (i)
can hold easily when all components of $W$ are discrete, which is not
uncommon in microeconometric applications of binary response models (see
e.g., \citet{komarova2013} and \citet{magnac2008}). In the presence of
continuous covariates, this condition becomes more restrictive.

For any real vector $u$, let $\left\Vert u\right\Vert _{E}=\sqrt{u^{\prime}u}
$ denote the Euclidean norm of $u$. To state a minimax lower bound, we first
define the following class of distributions.

\begin{definition}
For every $h\in \left( 0,1\right) $, let $\mathcal{P}(h,\mathcal{B}_{q})$
denote the class of distributions $F$ satisfying the following conditions:
(i) $b^{\ast }\in \mathcal{B}_{q}$, (ii) condition (\ref{margin condition})
holds, and (iii) there are constants $c_{u}>0$ and $c_{l}>0$ such that, for
any two vectors $\theta =\left( \alpha ,\beta ,\gamma \right) ,$ $\widetilde{%
\theta }=\left( \widetilde{\alpha },\widetilde{\beta },\widetilde{\gamma }%
\right) \in \{-1,1\}\times \Theta _{q}$ satisfying $\alpha =\widetilde{%
\alpha }$ and $\beta =\widetilde{\beta }$, it holds that%
\begin{equation}
c_{l}\left\Vert \theta -\widetilde{\theta }\right\Vert _{E}\leq \left\Vert
b_{\theta }-b_{\widetilde{\theta }}\right\Vert _{1}\leq c_{u}\left\Vert
\theta -\widetilde{\theta }\right\Vert _{E}.
\label{condition on the distribution of W}
\end{equation}
\end{definition}

The first two conditions in the definition of $\mathcal{P}(h,\mathcal{B}%
_{q}) $ have already been introduced before. The new condition (iii) imposes
that the Euclidean norm $\left\Vert \theta -\widetilde{\theta }\right\Vert
_{E}$ is equivalent to the $L_{1}$-norm $\left\Vert b_{\theta }-b_{%
\widetilde{\theta }}\right\Vert _{1}$ for two values $\theta $ and $%
\widetilde{\theta }$ that differ only in the components corresponding to the
auxiliary covariate coefficients. This condition is concerned with
restrictions on the distribution of the covariate vector $W$. The following
proposition gives sufficient conditions for verifying this norm equivalence
condition.

For any subset $J\subset \{1,...,p\}$, let $Z_{J}$ denote the $\left\vert
J\right\vert $-dimensional subvector of $Z\equiv (Z^{(1)},\ldots
,Z^{(p)})^{\prime }$ formed by keeping only those elements $Z^{(j)}$ with $%
j\in J$. Let $\mathcal{I}_{q}\equiv \cup _{\left( \beta ,\gamma \right) \in
\Theta _{q}}\text{Supp}(\widetilde{X}\beta +Z^{\prime }\gamma )$, where $%
\text{Supp}(V)$ denotes the support of the random variable $V$.

\begin{proposition}
\label{primitive assumptions on norm equivalence} Suppose that $s$ is fixed
and does not grow with sample size $n$. Assume that there are positive real
constants $L_{1}$, $L_{2}$ and $L_{3}$ such that (a) the distribution of $%
X_{0}$ conditional on $(\widetilde{X},Z)$ has a Lebesgue density that is
bounded above by $L_{1}$ and bounded below by $L_{1}^{-1}$ on $\mathcal{I}%
_{q}$, and (b) for any subset $J\subset \{1,...,p\}$ such that $\left\vert
J\right\vert \leq 2q$, $P\left( \left\Vert Z_{J}\right\Vert _{E}\leq
L_{2}\right) =1$ and the smallest eigenvalue of $E\left( Z_{J}Z_{J}^{\prime
}\right) $ is bounded below by $L_{3}$. Then Condition (iii) stated in 
\eqref{condition
on the distribution of W} holds with $c_{u}=L_{1}L_{2}$ and $%
c_{l}=(L_{1}L_{2})^{-1}L_{3}$.
\end{proposition}

Condition (a) in Proposition \ref{primitive assumptions on norm equivalence}
is mild. The first part of condition (b) holds with $L_{2}=\overline{L}\sqrt{%
2q}$ if $\max_{j\in \{1,...,p\}}\left\vert Z^{(j)}\right\vert \leq \overline{%
L}$ with probability 1 for some universal positive constant $\overline{L}$.
The second part of condition (b) is related to the sparse eigenvalue
assumption used in the high dimensional regression literature (see, e.g. %
\citet{Raskutti2011}). For example, suppose that $Z$ is a random vector with
mean zero and the covariance matrix $\Sigma $ whose $(i,j)$ component is $%
\Sigma _{i,j}=\rho ^{|i-j|}$ for some constant $\rho >0$. Then the smallest
eigenvalue of $\Sigma $ is bounded away from zero where the lower bound is
independent of the dimension $p$ (\citet[][p.\thinspace1384]{van2009}).
Thus, in this case, $E\left( Z_{J}Z_{J}^{\prime }\right) $ is bounded below
by that same lower bound.

We now state the result on the minimax lower bound for the predictive
performance of PRESCIENCE.

\begin{theorem}
\label{risk lower bound}Assume the parameter space $\Theta $ in (\ref%
{theta_q}) satisfies that there is a universal constant $\kappa >0$ such
that 
\begin{equation}
\max\nolimits_{j\in \{1,...,p\}}\left\vert \gamma _{j}\right\vert \leq
\kappa ,  \label{parameter bound}
\end{equation}%
where $\gamma _{j}$ denotes the $j$th component of $\gamma $. Suppose $p$
and $q$ are even numbers and $q<2p/3$. Let $\phi \equiv 0.71$. Then, for any
binary predictor $\widehat{b}:\mathcal{W}\mapsto \{0,1\}$, which is in the
set $\mathcal{B}_{q}$ and is constructed based on the data $\left(
Y_{i},W_{i}\right) _{i=1}^{n}$, we have that 
\begin{equation}
\sup\nolimits_{F\in \mathcal{P}(h,\mathcal{B}_{q})}E\left[ S_{q}^{\ast }-%
\widetilde{S}(\widehat{b})\right] \geq \frac{\phi qc_{l}\left( 1-\phi
\right) \left( 1-h\right) }{32nhc_{u}}\ln \left( \frac{p-q}{q/2}\right)
\label{risk bound}
\end{equation}%
for $h\in \left( 0,1\right) $, which is defined in \eqref{margin condition},
such that 
\begin{equation}
h\geq \left[ \frac{\phi \sqrt{q}\ln \left( \frac{p-q}{q/2}\right) }{8\sqrt{2}%
\kappa nc_{u}}\right] ^{1/2}.  \label{lower bound on h}
\end{equation}
\end{theorem}

For any estimator $\widehat{b}$ taking value in $\mathcal{B}_{q}$, Theorem %
\ref{risk lower bound} implies that, as long as $\mathcal{P}(h,\mathcal{B}%
_{q})$ is non-empty, there is some distribution $F$ under which the average
predictive risk $E\left[ S_{q}^{\ast }-\widetilde{S}(\widehat{b})\right] $
cannot be smaller than the lower bound term stated in \eqref{risk bound}.
Comparing the upper and lower bounds given by (\ref{finite sample upper
bound}) and (\ref{risk bound}), we can deduce conditions under which these
two bounds coincides in terms of rate of convergence such that the
PRESCIENCE approach is rate-optimal in the minimax sense. Suppose that $%
(s,c_{l},c_{u})$ are fixed and does not increase or decrease with $n$. Then
the risk lower bound is of order%
\begin{equation}
O\left( \frac{\left( 1-h\right) \ln p}{nh}\right) .
\label{order of lower bound rate}
\end{equation}%
Comparing \eqref{order of lower bound rate} to \eqref{sharper rate}
evaluated at $\vartheta =1$, we see that, if $h$ is also a universal
constant and $p$ grows at a polynomial or exponential rate in $n$, then the
upper and lower bound results induce the same convergence rate and hence the
PRESCIENCE approach is minimax rate-optimal. On the other hand, when
Condition \ref{margin condition for upper bound} holds with $\vartheta >1$,
the rate given by (\ref{sharper rate}) is slower than that given by (\ref%
{order of lower bound rate}) such that the convergence rate implied by the
risk lower bound need not be attained and therefore the PRESCIENCE approach
may not be rate-optimal.

\begin{remark}
The minimax rate optimality of PRESCIENCE is established under the
assumption that $s$ is fixed. Theorem \ref{risk lower bound} does not
provide a rate-optimal lower bound when $s$ diverges to infinity as $%
n\rightarrow \infty $, although it is a valid lower bound in any finite
sample. It is an interesting open question for future research to
investigate minimax optimality when $s\rightarrow \infty $.
\end{remark}

\begin{remark}
The assumption that $p$ and $q$ are even in Theorem \ref{risk lower bound}
is innocuous for the minimax rate-optimality result. This assumption is made
to invoke the known result (see Lemma 4 of \citet{Raskutti2011}) for the
lower bound on the complexity of the $\ell _{0} $-ball. When $p$ and/or $q$
is odd, the lower bound result still holds since we can always consider $%
\mathcal{P}(h,\mathcal{B}_{q}^{\prime })\subset \mathcal{P}(h,\mathcal{B}%
_{q})$, where $\mathcal{B}_{q}^{\prime }$ is a subspace of $\mathcal{B}_{q}$
for which the parameter vector $\theta $ is confined to a lower dimensional
space with dimension $p-1$ and/or $q-1$.
\end{remark}

\section{Implementation via Mixed Integer Optimization\label{implementation}}

We now present algorithms for solving the maximization problem (\ref{best
subset maximum score binary prediction rule}). It is straightforward to see
that solving (\ref{best subset maximum score binary prediction rule}) is the
same as solving%
\begin{equation*}
\max \left\{ \max\nolimits_{\left( \beta ,\gamma \right) \in \Theta _{q}}%
\text{ }S_{n}(1,\beta ,\gamma ),\max\nolimits_{\left( \beta ,\gamma \right)
\in \Theta _{q}}\text{ }S_{n}(-1,\beta ,\gamma )\right\} .
\end{equation*}%
In what follows, we focus on solving the sub-problem%
\begin{equation}
\max\nolimits_{\left( \beta ,\gamma \right) \in \Theta _{q}}\text{ }%
S_{n}(1,\beta ,\gamma )  \label{sub-problem}
\end{equation}%
because the other case corresponding to $\alpha =-1$ can be solved by
replacing the value of $X_{0i}$ with that of $-X_{0i}$ and then applying the
same solution method as developed for the case (\ref{sub-problem}).

By (\ref{augmented model}) and noting that $Y_{i}\in\{0,1\}$, solving the
problem (\ref{sub-problem}) amounts to solving%
\begin{equation}
\max\nolimits_{\left( \beta,\gamma\right)
\in\Theta_{q}}n^{-1}\sum\nolimits_{i=1}^{n}\left[ \left( 1-Y_{i}\right)
+\left( 2Y_{i}-1\right) 1\{X_{0i}+\widetilde{X}_{i}^{\prime}\beta+Z_{i}^{%
\prime}\gamma \geq0\}\right] .  \label{max score}
\end{equation}
We assume that the parameter space $\Theta$ is bounded and takes the
polyhedral form: 
\begin{equation*}
\Theta=\{\left( \beta,\gamma\right) \in\mathbb{R}^{k+p}:A_{1}\beta
+A_{2}\gamma\leq B\}
\end{equation*}
for some real constant matrices $A_{1}$ and $A_{2}$ and some real constant
vector $B$. Let 
\begin{equation}
\mathbf{C} \equiv\prod\nolimits_{j=1}^{p}\left[ \underline{\gamma}_{j},%
\overline{\gamma}_{j}\right]  \label{space C}
\end{equation}
denote the smallest cube containing all values of $\gamma$ in the pair $%
\left( \beta,\gamma\right) $ confined by $\Theta$. Writing $\gamma
=(\gamma_{1},...,\gamma_{p})$, we have that, if $\left( \beta,\gamma\right)
\in\Theta$, then $-\infty<\underline{\gamma}_{j}\leq\gamma_{j}\leq \overline{%
\gamma}_{j}<\infty$ for $j\in\{1,...,p\}$. Let 
\begin{equation}
M_{i}\equiv\max_{(\beta,\gamma)\in\Theta}\left\vert X_{0i}+\widetilde{X}%
_{i}^{\prime}\beta+Z_{i}^{\prime}\gamma\right\vert \text{ for }i\in
\{1,...,n\}.  \label{Mi}
\end{equation}

Our implementation builds on the method of mixed integer optimization (in
particular, \citet{bertsimas2016}, \citet{Florios:Skouras:08}, and %
\citet{Kitagawa:Tetenov:2015}) and present two alternative solution methods
that complement each other. The values $\left( M_{i}\right) _{i=1}^{n}$ can
be computed by formulating the maximization problem in (\ref{Mi}) as linear
programming problems, which can be easily and efficiently solved by modern
numerical software. Hence these values can be computed and stored beforehand
as inputs to the algorithms that are used to solve the MIO problems
described below.

\subsection{Method 1}

Our first solution method is based on an equivalent reformulation of the
maximization problem (\ref{max score}) as the following constrained mixed
integer optimization (MIO) problem:%
\begin{align}
& \max_{(\beta ,\gamma )\in \mathbf{\Theta }%
,d_{1},...,d_{n},e_{1},...,e_{p}}n^{-1}\sum\nolimits_{i=1}^{n}\left[ \left(
1-Y_{i}\right) +\left( 2Y_{i}-1\right) d_{i}\right]  \label{constrained MIO}
\\
& \text{subject to}  \notag \\
& \left( d_{i}-1\right) M_{i}\leq X_{0i}+\widetilde{X}_{i}^{\prime }\beta
+Z_{i}^{\prime }\gamma <d_{i}(M_{i}+\delta ),\text{ }i\in \{1,...,n\},
\label{constraint on di} \\
& e_{j}\underline{\gamma }_{j}\leq \gamma _{j}\leq e_{j}\overline{\gamma }%
_{j},\text{ }j\in \{1,...,p\},  \label{selection constraint} \\
& \dsum\nolimits_{j=1}^{p}e_{j}\leq q,  \label{cardinality constraint} \\
& d_{i}\in \{0,1\},\text{ }i\in \{1,...,n\},  \label{indicator di} \\
& e_{j}\in \{0,1\},\text{ }j\in \{1,...,p\},  \label{selection indicator ej}
\end{align}%
where $\delta $ is a given small and positive real scalar (e.g. $\delta
=10^{-6}$ as in our numerical study).

We now explain the equivalence between (\ref{max score}) and (\ref%
{constrained MIO}). Given $(\beta ,\gamma )$, the inequality constraints (%
\ref{constraint on di}) and the dichotomization constraints (\ref{indicator
di}) enforce that $d_{i}=1\{X_{0i}+\widetilde{X}_{i}^{\prime }\beta
+Z_{i}^{\prime }\gamma \geq 0\}$ for $i\in \{1,...,n\}$. Therefore,
maximizing the objective function in (\ref{max score}) for $(\beta ,\gamma
)\in \Theta $ subject to the constraints (\ref{constraint on di}) and (\ref%
{indicator di}) is equivalent to solving the problem (\ref{max score}) using
all covariates. This part of formulation is similar to the MIO formulation
used by \citet{Kitagawa:Tetenov:2015} for solving the maximum score type
estimation problems without the variable selection constraint.

Following \citet{bertsimas2016}, we implement the best subset variable
selection feature through the additional constraints (\ref{selection
constraint}), (\ref{cardinality constraint}) and (\ref{selection indicator
ej}). The on-off constraints (\ref{selection constraint}) and (\ref%
{selection indicator ej}) ensure that, whenever $e_{j}=0$, the auxiliary
covariate $Z_{j}$ is excluded in the resulting PRESCIENCE. Finally, the
cardinality constraint $\left\Vert \gamma\right\Vert _{0}\leq q$ is enacted
through the constraint (\ref{cardinality constraint}), which restricts the
maximal number of the binary controls $e_{j}$ that can take value unity. 

Modern numerical optimization solvers such as CPLEX, Gurobi, MOPS, Mosek and
Xpress-MP can be used to effectively solve the MIO formulations of the
PRESCIENCE problem. Most of the solution algorithms employed by the MIO
solvers can be viewed as complex and advanced refinements of the well-known
branch-and-bound method for solving MIO problems.\footnote{%
See Online Appendix \ref{bb-method} for further details of the
branch-and-bound method.} Along the branch-and-bound solution process, we
can keep track of two important values: the best upper and lower bounds on
the objective value of the MIO problem (\ref{constrained MIO}). The best
lower bound corresponds to the objective function evaluated at the incumbent
solution, which is the best feasible solution discovered so far. The best
upper bound can be deduced by taking the maximum of the optimal objective
values of all the linear programming relaxation formulations of the
branching MIO sub-problems that have been solved so far. Let $MIO\_gap$
denote the difference between these two bounds. Note that the incumbent
solution becomes optimal when the $MIO\_gap$ value reduces to zero.

We can use the $MIO\_gap$ value to solve for the $\varepsilon $-level
PRESCIENCE introduced in Section \ref{covariate selection}. To see this,
consider an early termination rule by which the solution algorithm is
terminated whenever $MIO\_gap\leq \varepsilon $ where $\varepsilon $ is a
given tolerance level. Let $(\widehat{\beta },\widehat{\gamma })$ be the
incumbent solution upon termination of the MIO solver. Because $(\widehat{%
\beta },\widehat{\gamma })$ is in the feasible solution set of the problem (%
\ref{constrained MIO}), by constraints (\ref{selection constraint}) and (\ref%
{cardinality constraint}), we have that $\left\Vert \widehat{\gamma }%
\right\Vert _{0}\leq q$ . Moreover, by constraints (\ref{constraint on di})
and (\ref{indicator di}), we have that $\widehat{d}_{i}=1\{X_{0i}+\widetilde{%
X}_{i}^{\prime }\widehat{\beta }+Z_{i}^{\prime }\widehat{\gamma }\geq 0\}$
for $i\in \{1,...,n\}$ so that $S_{n}(1,\widehat{\beta },\widehat{\gamma })$
is equal to the objective function in (\ref{constrained MIO}) evaluated at $(%
\widehat{d}_{1},...,\widehat{d}_{n})$. Since (\ref{max score}) and (\ref%
{constrained MIO}) are equivalent maximization problems, it thus follows
from the construction of $MIO\_gap$ value that%
\begin{equation}
S_{n}(1,\widehat{\beta },\widehat{\gamma })\geq \max\nolimits_{\left( \beta
,\gamma \right) \in \Theta _{q}}\text{ }S_{n}(1,\beta ,\gamma )-MIO\_gap.%
\text{ }  \label{ineq}
\end{equation}%
Given the termination condition, we can therefore see that%
\begin{equation}
S_{n}(1,\widehat{\beta },\widehat{\gamma })\geq \max\nolimits_{\left( \beta
,\gamma \right) \in \Theta _{q}}\text{ }S_{n}(1,\beta ,\gamma )-\varepsilon ,
\label{approximate MIO solution}
\end{equation}%
which yields an approximately optimal solution with the optimization
tolerance level $\varepsilon $\ for the problem (\ref{sub-problem}).

The PRESCIENCE can also be solved by incorporating the constraints (\ref%
{selection constraint}), (\ref{cardinality constraint}) and (\ref{selection
indicator ej}) in the MIO formulation of \citet{Florios:Skouras:08} for
solving the maximum score estimation problem. We now present this
alternative MIO formulation below.

\subsection{Method 2}

Consider the constrained maximization problem:%
\begin{equation}
\max\nolimits_{\left( \beta ,\gamma \right) \in \Theta
_{q}}n^{-1}\sum\nolimits_{i=1}^{n}1\left\{ \left( 2Y_{i}-1\right) (X_{0i}+%
\widetilde{X}_{i}^{\prime }\beta +Z_{i}^{\prime }\gamma )\geq 0\right\} .
\label{Florios-Skouras max score formulation}
\end{equation}%
The problem (\ref{Florios-Skouras max score formulation}) without the
constraint $\left\Vert \gamma \right\Vert _{0}\leq q$ reduces to the type of
maximum score estimation problem studied by \citet{Florios:Skouras:08}. The
objective function in (\ref{Florios-Skouras max score formulation})
coincides with that in (\ref{sub-problem}) with probability 1 as long as the
sum $X_{0}+\widetilde{X}^{\prime }\beta +Z^{\prime }\gamma $ is continuously
distributed. This condition holds provided that the distribution of $X_{0}$
conditional on $(\widetilde{X},Z)$ is continuous. With such a continuous
covariate, we can also solve (\ref{sub-problem}) by solving the following
MIO formulation of (\ref{Florios-Skouras max score formulation}):%
\begin{align}
& \max_{(\beta ,\gamma )\in \mathbf{\Theta }%
,d_{1},...,d_{n},e_{1},...,e_{p}}n^{-1}\sum\nolimits_{i=1}^{n}d_{i}\text{ }
\label{Florios-Skouras MIO} \\
& \text{subject to the constraints (\ref{selection constraint}), (\ref%
{cardinality constraint}), (\ref{indicator di}), (\ref{selection indicator
ej}), and}  \notag \\
& \left( 1-2Y_{i}\right) (X_{0i}+\widetilde{X}_{i}^{\prime }\beta
+Z_{i}^{\prime }\gamma )\leq M_{i}\left( 1-d_{i}\right) ,\text{ }i\in
\{1,...,n\}.  \label{sign matching constraints}
\end{align}

\citet{Florios:Skouras:08} showed that maximizing the objective function in (%
\ref{Florios-Skouras MIO}) for $(\beta ,\gamma )\in \Theta $ subject to the
constraints (\ref{indicator di}) and (\ref{sign matching constraints}) is
equivalent to solving the problem (\ref{Florios-Skouras max score
formulation}) using all covariates. This can be seen from the fact that the
objective function of the MIO problem (\ref{Florios-Skouras MIO}) is
strictly increasing in $d_{i}$ so that, given $(\beta ,\gamma )$, it is
optimal to set $d_{i}=1\{\left( 2Y_{i}-1\right) (X_{0i}+\widetilde{X}%
_{i}^{\prime }\beta +Z_{i}^{\prime }\gamma )\geq 0\}$ under the constraints (%
\ref{indicator di}) and (\ref{sign matching constraints}). Along similar
arguments to those discussed for the problem (\ref{constrained MIO}), it is
also straightforward to verify that the variable selection constraint $%
\left\Vert \gamma \right\Vert _{0}\leq q$ is imposed through the constraints
(\ref{selection constraint}), (\ref{cardinality constraint}) and (\ref%
{selection indicator ej}). Therefore, the maximization problems (\ref%
{Florios-Skouras max score formulation}) and (\ref{Florios-Skouras MIO}) are
equivalent.

\subsection{Tightening the Parameter Space as a Warm Start to the MIO
Formulation of the PRESCIENCE Problem\label{sec:para-space}}

The MIO formulations (\ref{constrained MIO}) and (\ref{Florios-Skouras MIO})
depend on the specification of the parameter space $\Theta$. It is well
known that use of a good and tighter parameter space can strengthen the
performance of a global optimization procedure. Given an initial
specification of $\Theta $, we propose below a data driven approach to
refine the parameter space.

Recall that $W_{i}=(X_{0i},\widetilde{X}_{i},Z_{i})$ is the entire covariate
vector. Let $\widetilde{W}_{i}$ denote the vector $(\widetilde{X}_{i},Z_{i})$%
. For $i\in \{1,...,n\}$, let $\widehat{P}_{i}$ be an estimate of ${P}%
_{i}\equiv P(Y_{i}=1|W_{i})$. Define the following sets recursively:

\begin{equation*}
\underline{\Theta }_{1}\equiv \Theta \mathbf{,}\text{ }\overline{\Theta }%
_{1}\equiv \left\{ \left( t_{1},...,t_{k+p}\right) \in \Theta :t_{1}\geq 
\widehat{l}_{1}\right\}
\end{equation*}%
and, for $m\in \{2,...,k+p\}$, 
\begin{align}
\underline{\Theta }_{m}& \equiv \left\{ \left( t_{1},...,t_{k+p}\right) \in
\Theta :\widehat{l}_{s}\leq t_{s}\leq \widehat{u}_{s}\text{ for }s\in
\{1,...,m-1\}\right\} ,  \label{theta_lowerbar} \\
\overline{\Theta }_{m}& \equiv \left\{ \left( t_{1},...,t_{k+p}\right) \in 
\underline{\Theta }_{m}:t_{m}\geq \widehat{l}_{m}\right\} .
\label{theta_upperbar}
\end{align}%
where, for $j\in \{1,...,k+p\}$, the quantities $\widehat{l}_{j}$ and $%
\widehat{u}_{j}$ are defined respectively by%
\begin{align}
& \widehat{l}_{j}\equiv \min\limits_{t\in \underline{\Theta }_{j}}t_{j}\text{
subject to }  \label{lower bound} \\
& (X_{0i}+\widetilde{W}_{i}^{\prime }t)(\widehat{P}_{i}-0.5)\geq 0\text{ for 
}i\in \{1,...,n\}.  \label{sign constraints} \\[5pt]
& \widehat{u}_{j}\equiv \max\limits_{t\in \overline{\Theta }_{j}}t_{j}\text{
subject to the constraints (\ref{sign constraints}).}  \label{upper bound}
\end{align}%
If the binary outcome $Y_{i}$ is generated from the model specified by (\ref%
{binary response model}) and (\ref{median independence}) and the conditional
probability $P(Y=1|W)$ is nonparametrically estimated, the interval $[%
\widehat{l}_{j},\widehat{u}_{j}]$ is a nonparametric estimate of the
identified set for the $j$th component of the parameter vector $t=(\beta
,\gamma )$. In this case, the sign-matching constraints (\ref{sign
constraints}) can be regarded as the empirical counterparts of the
inequalities stated in the set 
\begin{equation*}
\left\{ t\in \Theta :(X_{0i}+\widetilde{W}_{i}^{\prime
}t)(P(Y_{i}=1|W_{i})-0.5)\geq 0\text{ almost surely}\right\} ,
\end{equation*}%
which contains those $t$ values that are observationally equivalent to the
true data generating parameter value 
\citep[see, e.g.][]{komarova2013,
chen2015}. Our procedures for computing $\widehat{l}_{j}$ and $\widehat{u}%
_{j}$ are modified versions of \citet[p. 62]{horowitz1998}'s linear
programming formulations of the identified bounds on the parameter
components. The formulations (\ref{lower bound}) and (\ref{upper bound})
differ from those of Horowitz in that we further tighten the domain of $t$
in these optimization problems by exploiting the information of the upper
and lower bound values that have been solved so far.

When the covariate vector is of high dimension, nonparametric estimation of $%
P(Y=1|W)$ would suffer from the curse of dimensionality problem. In this
paper, we consider estimating this conditional probability by parametric
methods such as the logit or probit approach. In Monte Carlo experiments and
an empirical application, we estimate $({P}_{i})_{i=1}^{n}$ by the fitted
choice probabilities from the logit regression of $Y$ on all the covariates.

Noting that the parametric model for estimating $P(Y=1|W)$ may be
misspecified, for $\tau \geq 1$, we construct a conservative space $\widehat{%
\Theta }\left( \tau \right) $, which is a $\tau $-enlargement of the space $%
\prod\nolimits_{j=1}^{k+p}\left[ \widehat{l}_{j},\widehat{u}_{j}\right] $ as
given below:%
\begin{equation*}
\widehat{\Theta }\left( \tau \right) \equiv \left\{ t\in \Theta :-\tau
\left( \left\vert \widehat{l}_{j}\right\vert \vee \left\vert \widehat{u}%
_{j}\right\vert \right) \leq t_{j}\leq \tau \left( \left\vert \widehat{l}%
_{j}\right\vert \vee \left\vert \widehat{u}_{j}\right\vert \right) \text{
for }j\in \{1,...,k+p\}\right\} .
\end{equation*}

We can solve the MIO problems (\ref{constrained MIO}) and (\ref%
{Florios-Skouras MIO}) with the refined parameter space $\widehat{\Theta }%
\left( \tau \right) $ in place of the original space $\Theta $.\footnote{%
Both the covering cube $\mathbf{C}$ and the quantities $\left( M_{i}\right)
_{i=1}^{n}$ depend on the input parameter space. Hence, for the warm-start
formulations of (\ref{constrained MIO}) and (\ref{Florios-Skouras MIO}),
these objects are also computed under the refined space $\widehat{\Theta }%
\left( \tau \right) $.} Using the terminology used in \citet{bertsimas2016},
we shall refer to these refined MIO representations as the \emph{warm-start}
MIO formulations of the PRESCIENCE problem. The value of $\tau $ is treated
as a tuning parameter for solving the warm-start MIO problems. The original
formulations (\ref{constrained MIO}) and (\ref{Florios-Skouras MIO}) based
on the space $\Theta $ are referred to as the \emph{cold-start} MIO
formulations.

Computation of the refined space $\widehat{\Theta }\left( \tau \right) $
requires solving $2\left( k+p\right) $ simple linear programming problems.
This task can be done very efficiently even when $p$ is relatively large. On
the other hand, the space $\widehat{\Theta }\left( \tau \right) $ is not
always constructible since the problems (\ref{lower bound}) and (\ref{upper
bound}) may not admit any feasible solution. This may occur due to the
misspecification issue of using parametric choice probability estimates.
Alternatively, it can also occur when the postulated binary response model
specified by (\ref{binary response model}) and (\ref{median independence})
itself is misspecified. As illustrated by Monte Carlo simulations and a real
data application in Online Appendices \ref{simulations of the cold and warm
start methods} and \ref{emp:linear:spec}, when the refined space $\widehat{%
\Theta }\left( \tau \right) $ is available, solving the warm-start MIO
formulations can be computationally far more efficient than solving their
corresponding cold-start versions.

We conclude this subsection by commenting that our warm-start approach does
not work well when $p$ is greater than $n$. In this case, irrespective of
the knowledge of the true choice probabilities, the dimension of the vector
of unknown coefficients is larger than the number of inequalities given by
the constraints (\ref{sign constraints}) such that these constraints may
become ineffective for tightening the original parameter bounds. It is a
topic for future research how to devise a good wart-start option for the
high-dimensional setup.

\section{Simulation Study\label{simulation study}}

In this section, we study the performance of the PRESCIENCE method in Monte
Carlo experiments. Throughout this paper, we used the MATLAB implementation
of the Gurobi Optimizer to solve the MIO problems. Moreover, all numerical
computations were done on a desktop PC (Windows 7) equipped with 32 GB RAM
and a CPU processor (Intel i7-5930K) of 3.5 GHz.\footnote{%
The MATLAB codes for implementing the PRESCIENCE approach are available from
the authors via the website %
\url{https://github.com/LeyuChen/Best-Subset-Binary-Prediction}. This
implementation requires the Gurobi Optimizer, which is freely available for
academic purposes.}

Let $V=(V_{1},...,V_{p+1})$ be a multivariate normal random vector with mean
zero and covariance matrix $\Sigma $ with its element $\Sigma _{i,j}=\left(
0.25\right) ^{\left\vert i-j\right\vert }$. The binary outcome is generated
according to the following setup:%
\begin{equation*}
Y=1\{W^{\prime }\theta ^{\ast }\geq \sigma (W)\xi \},
\end{equation*}%
where $\theta ^{\ast }$ denotes the value of the true data generating
parameter vector, $W=(X,Z)$ is a $(p+2)$ dimensional covariate vector with
the focus covariates $X=(X_{0},\widetilde{X})=(V_{1},1)$ and the auxiliary
covariates $Z=(V_{2},...,V_{p+1})$, and $\xi $ is a $N(0,1)$ random variate
independent of $V$. We set $\theta _{1}^{\ast }=1$, $\theta _{2}^{\ast }=0$,
and $\theta _{j}^{\ast }=0$ for $j\in \{4,...,p+2\}.$ The coefficient $%
\theta _{3}^{\ast }$ is chosen to be non-zero such that, among the $p$
auxiliary covariates, only the variable $Z_{1}$ is relevant in the data
generating processes (DGP).

We consider the following two specifications for $\theta _{3}^{\ast }$ and $%
\sigma (W)$: 
\begin{eqnarray*}
&&\text{DGP(i) : }\theta _{3}^{\ast }=-0.35\text{ and }\sigma (W)=0.25. \\
&&\text{DGP(ii) : }\theta _{3}^{\ast }=-1.5\text{ and }\sigma (W)=0.25\left(
1+2\left( V_{1}+V_{2}\right) ^{2}+\left( V_{1}+V_{2}\right) ^{4}\right) .
\end{eqnarray*}

As before, the parameter vector $\theta $ in (\ref{predictor}) is decomposed
as $\theta =(\alpha ,\beta ,\gamma )$ where $\alpha $, $\beta $ and $\gamma $
are coefficients associated with $X_{0}$, $\widetilde{X}$ and $Z$,
respectively. The parameter space for the PRESCIENCE approach is specified
to be%
\begin{equation}
\{\left( \alpha ,\beta ,\gamma \right) \in \mathbb{R}^{p+2}:\alpha =1,(\beta
,\gamma )=(\beta ,\gamma _{1},...,\gamma _{p})\in \lbrack -10,10]^{p+1}\}
\label{PRESCIENCE parameter space}
\end{equation}%
over which we compute PRESCIENCE via solving its corresponding MIO problem.

There were $100$ simulation repetitions in each Monte Carlo experiment. For
each simulation repetition, we generated a training sample of $n$
observations for estimating the coefficients $\theta $ and a validation
sample of $5000$ observations for evaluating the out-of-sample predictive
performance. The training sample size $n$ was set to be 100 for DGP(i) and
50 for DGP(ii). For each DGP setup, we performed simulations with both the
low and high dimensional covariate configurations. For the low dimensional
case, we set $p=10$ for both DGP(i) and (ii). For the high dimensional case,
we set $p=200$ for DGP(i) and $p=60$ for DGP(ii).

We considered the following class of prediction methods:%
\begin{equation}
\mathcal{M}=\{\{\text{PRESCIENCE}\left( q\right) \text{, }q\in \{1,2,3\}\},%
\text{PRE\_CV},\text{logit\_lasso,probit\_lasso}\}\text{,}
\label{prediction methods}
\end{equation}%
where PRESCIENCE$\left( q\right) $ denotes the PRESCIENCE approach with a
cardinality bound $q$ imposed on the auxiliary covariates, PRE\_CV denotes
the PRESCIENCE approach using a data driven value of $q\in \{1,2,3\}$ via
the 5-fold cross validation procedure, and logit\_lasso and probit\_lasso
respectively denote the $\ell _{1}$-penalized logit and probit maximum
likelihood estimation (MLE) approaches \citep[see
e.g.][]{friedman2010}. Throughout this simulation study, we employed the
cold-start MIO formulation (\ref{constrained MIO}) to solve the PRESCIENCE
problems. For the simulation experiment with $p<n$, we computed the exact
solution to each PRESCIENCE problem. For the high dimensional case with $p>n$%
, we solved for the PRESCIENCE solution with the tolerance level $%
\varepsilon $ specified according to the rule%
\begin{equation}
\varepsilon =\min \{0.05,0.5\sqrt{\ln (p\vee n)/n}\}.
\label{early stopping rule}
\end{equation}%
Note that this early termination rule is compatible with the order of
magnitude stated in the condition (\ref{epsilon}) for the convergence rate
result \eqref{theorem1-expectation}. For the logit\_lasso and probit\_lasso
approaches, we used the MATLAB function \textbf{lassoglm} to implement these
two penalized MLE approaches for which we calibrated the lasso penalty
parameter value over a sequence of 100 values via the 10-fold cross
validation procedure. We used the default setup of \textbf{lassoglm} for
constructing this tuning sequence among which we made the following three
choices, $\left\{ \lambda _{\min },\lambda _{1se},\lambda _{2se}\right\} $,
of the penalty parameter value. To describe these, let $\{I_{j}:j=1,\ldots
,K\}$ be the partition of data, where $K=10$, and let $\hat{L}(I_{j},\lambda
)$ denote the minus log likelihood function evaluated using data in $I_{j}$
but estimating the model using data in $\bigcup\nolimits_{i\in
\{1,...,K\}\backslash \{j\}}I_{i}$ with a given penalty parameter value $%
\lambda $. The value $\lambda _{\min }$ refers to the $\lambda $ value that
minimized the mean cross validated deviances ($K^{-1}\sum_{j=1}^{K}\hat{L}%
(I_{j},\lambda )$), whereas $\lambda _{1se}$ and $\lambda _{2se}$
respectively denote the largest penalty parameter values whose corresponding
mean cross validated deviances still fall within the one- and two-standard
errors of $K^{-1}\sum_{j=1}^{K}\hat{L}(I_{j},\lambda _{\min })$.\footnote{%
Here, standard errors are computed over the 10 cross-validation folds.
Choice of the lasso tuning parameter based on $\lambda _{1se}$ is also known
as the "one-standard-error" rule, which is commonly employed in the
statistical learning literature \citep{Hastie2009}.}

For each $m\in $ $\mathcal{M}$, let $\widehat{\theta }(m)$ denote the
coefficients computed under the prediction method $m$. Let $in\_Score$
denote the average of the in-sample objective values $S_{n}(\widehat{\theta }%
(m))$ over all the simulation repetitions. In each simulation repetition, we
approximated the out-of-sample objective value $S(\widehat{\theta }(m))$
using the generated validation sample. Let $out\_Score$ denote the average
of $S(\widehat{\theta }(m))$ over all the simulation repetitions. It is
straightforward to see that the theoretically best prediction rule $b^{\ast
}(w)$ in this simulation design takes the form $b^{\ast }(w)=1\left\{
w^{\prime }\theta ^{\ast }\geq 0\right\} $. Hence, we also assess the
predictive performance of a given prediction method $m\in \mathcal{M}$ by
its relative score, which is ratio of the score evaluated at $\widehat{%
\theta }(m)$ over that evaluated at $\theta ^{\ast }$. Let $in\_RS$ and $%
out\_RS$ respectively denote the average of in-sample relative scores $S_{n}(%
\widehat{\theta }(m))/$ $S_{n}(\theta ^{\ast })$ and that of out-of-sample
relative scores $S(\widehat{\theta }(m))/S(\theta ^{\ast })$ over all the
simulation repetitions.

We also examine the variable selection performance of the prediction method.
We say that a variable $Z_{j}$ is effectively selected under the prediction
method $m$ if and only if the magnitude of $\widehat{\theta }_{j+2}(m)$ is
larger than a small tolerance level (e.g. $10^{-6}$ as used in our numerical
study) which is distinct from zero in numerical computation. Let $Corr\_sel$
be the proportion of the auxiliary covariate $Z_{1}$ being effectively
selected. Let $Orac\_sel$ be the proportion of obtaining an oracle variable
selection outcome where, among all the auxiliary covariates, $Z_{1}$ was the
only one that was effectively selected. Let $Num\_irrel$ denote the average
number of effectively selected auxiliary covariates whose true DGP
coefficients are zero.

\subsection{Simulation Results for the DGP(i) Design}

We now present the simulation results under the setup of DGP(i). First, we
report the computational performance of our MIO solution algorithm to the
PRESCIENCE problems. Table \ref{tab1} gives the summary statistics of the
MIO computation time in CPU seconds across simulation repetitions. From this
table, we can see that the MIO problems for the PRESCIENCE computation were
solved very efficiently in the DGP(i) simulations where the number of the
auxiliary covariates could be the double of the sample size yet the maximum
computation time was only around 5 minutes. It is also interesting to note
that the PRESCIENCE computation time was not monotone in $q$. This feature
might be due to the branching strategy heuristics of the MIO
branch-and-bound solution algorithms. 
\begin{table}[tbph]
\caption{PRESCIENCE computation time (CPU seconds) under DGP(i)}
\label{tab1}
\begin{center}
\begin{tabular}{c|ccc|ccc}
\hline\hline
& \multicolumn{3}{|c|}{$p=10$} & \multicolumn{3}{|c}{$p=200$} \\ \hline
$q$ & 1 & 2 & 3 & 1 & 2 & 3 \\ \hline
mean & 0.69 & 1.10 & 0.51 & 7.11 & 20.37 & 1.96 \\ 
min & 0.01 & 0.01 & 0.01 & 0.31 & 0.16 & 0.07 \\ 
median & 0.35 & 0.38 & 0.33 & 3.76 & 2.09 & 0.99 \\ 
max & 3.51 & 24.56 & 5.42 & 68.04 & 362.7 & 15.95 \\ \hline
\end{tabular}%
\end{center}
\end{table}

We next turn to the statistical performance of the binary prediction method.
In Tables \ref{tab2} and \ref{tab3}, we compare the aforementioned
predictive and variable selection performance measures for the various
prediction approaches given in (\ref{prediction methods}). As shown in these
two tables, regardless of $p$, the in-sample fit in terms of $in\_Score$ and 
$in\_RS$ for the PRESCIENCE($q$) method increased with $q$. This finding is
expected because the in-sample objective function (\ref{empirical objective
function}) is monotone in $q$ by design for the PRESCIENCE approach.
Nonetheless, both tables indicate that $out\_Score$ and $out\_RS$ also
declined as $q$ increased, thus resounding with the known issue that
in-sample overfitting may result in poor out-of-sample performance. When the
true number of effective auxiliary covariates is unknown, one can choose the 
$q$ value that maximizes the mean cross validated score. From Tables \ref%
{tab2} and \ref{tab3}, we find that the PRE\_CV approach indeed balanced
well the in-sample and out-of-sample predictive performances. Moreover, its
predictive performance measures were also comparable to those given by the
logit\_lasso and probit\_lasso approaches.

We now discuss the variable selection results. Table \ref{tab2} indicates
that all the prediction approaches had high $Corr\_sel$ rates and hence were
capable of effectively selecting the relevant covariate $Z_{1}$. However,
the good performance in the $Corr\_sel$ criterion may arise at the risk of
overfitting. Therefore, we also have to take into account the performance in
excluding irrelevant auxiliary covariates. The simulation design implies
that the case with $q=1$ is the most parsimonious PRESCIENCE setup that
correctly specifies the number of effective auxiliary covariates in the DGP.
Therefore, it is not surprising that PRESCIENCE(1) performed the best in
terms of $Num\_irrel$. We note that the PRE\_CV approach also performed very
well in excluding the irrelevant variables. In fact, for the PRESCIENCE
based approaches, only PRESCIENCE(1) and PRE\_CV could yield a non-zero
probability of inducing an oracle variable selection outcome. For the
penalized MLE approaches, the logit\_lasso and probit\_lasso coupled with
the larger penalty parameter value $\lambda _{2se}$ also performed well in
terms of $Orac\_Sel$ and $Num\_irrel$, albeit at the cost of a slight
reduction of the out-of-sample predictive performances.

We also observe a similar pattern in the results of the DGP(i) setup with $%
p=200$. From Table \ref{tab3}, we find that the PRE\_CV approach also
balanced very well the requirement for including the relevant but excluding
the irrelevant variables. It is also noted that the logit\_lasso and
probit\_lasso approaches in this high dimensional simulation setup tended to
select more irrelevant variables than the PRESCIENCE approach, hence
suffering from a larger extent of overfitting. 
\begin{table}[tbph]
\caption{Comparison of prediction methods under DGP(i) with $p=10$}
\label{tab2}
\begin{center}
\begin{tabular}{c|ccc|c|ccc|ccc}
\hline\hline
{\small method} & \multicolumn{3}{|c|}{{\small PRESCIENCE(}$q${\small )}} & 
{\small PRE\_CV} & \multicolumn{3}{|c}{\small logit\_lasso} & 
\multicolumn{3}{|c}{\small probit\_lasso} \\ 
& ${\small q=1}$ & ${\small q=2}$ & ${\small q=3}$ &  & ${\small \lambda }%
_{\min }$ & ${\small \lambda }_{1se}$ & ${\small \lambda }_{2se}$ & ${\small %
\lambda }_{\min }$ & ${\small \lambda }_{1se}$ & ${\small \lambda }_{2se}$
\\ \hline
${\small Corr\_sel}$ & {\small 0.93} & {\small 0.99} & {\small 1} & {\small %
0.97} & {\small 1} & {\small 1} & {\small 0.94} & {\small 1} & {\small 1} & 
{\small 0.95} \\ 
${\small Orac\_sel}$ & {\small 0.93} & {\small 0} & {\small 0} & {\small 0.51%
} & {\small 0} & {\small 0.18} & {\small 0.45} & {\small 0} & {\small 0.21}
& {\small 0.48} \\ 
${\small Num\_irrel}$ & {\small 0.07} & {\small 1.01} & {\small 1.99} & 
{\small 0.71} & {\small 5.03} & {\small 2.02} & {\small 0.97} & {\small 4.75}
& {\small 1.87} & {\small 0.9} \\ 
${\small in\_Score}$ & {\small 0.948} & {\small 0.964} & {\small 0.974} & 
{\small 0.960} & {\small 0.947} & {\small 0.930} & {\small 0.919} & {\small %
0.944} & {\small 0.928} & {\small 0.917} \\ 
${\small in\_RS}$ & {\small 1.028} & {\small 1.046} & {\small 1.058} & 
{\small 1.042} & {\small 1.028} & {\small 1.009} & {\small 0.997} & {\small %
1.024} & {\small 1.007} & {\small 0.995} \\ 
${\small out\_Score}$ & {\small 0.904} & {\small 0.901} & {\small 0.898} & 
{\small 0.903} & {\small 0.904} & {\small 0.905} & {\small 0.899} & {\small %
0.904} & {\small 0.904} & {\small 0.898} \\ 
${\small out\_RS}$ & {\small 0.982} & {\small 0.979} & {\small 0.976} & 
{\small 0.981} & {\small 0.983} & {\small 0.983} & {\small 0.977} & {\small %
0.983} & {\small 0.983} & {\small 0.976} \\ \hline
\end{tabular}%
\end{center}
\end{table}

\begin{table}[tbph]
\caption{Comparison of prediction methods under DGP(i) with $p=200$}
\label{tab3}
\begin{center}
\begin{tabular}{c|ccc|c|ccc|ccc}
\hline\hline
{\small method} & \multicolumn{3}{|c|}{{\small PRESCIENCE(}$q${\small )}} & 
{\small PRE\_CV} & \multicolumn{3}{|c}{\small logit\_lasso} & 
\multicolumn{3}{|c}{\small probit\_lasso} \\ 
& ${\small q=1}$ & ${\small q=2}$ & ${\small q=3}$ &  & ${\small \lambda }%
_{\min }$ & ${\small \lambda }_{1se}$ & ${\small \lambda }_{2se}$ & ${\small %
\lambda }_{\min }$ & ${\small \lambda }_{1se}$ & ${\small \lambda }_{2se}$
\\ \hline
${\small Corr\_sel}$ & {\small 0.78} & {\small 0.88} & \multicolumn{1}{c|}%
{\small 0.89} & {\small 0.86} & {\small 0.99} & {\small 0.9} & {\small 0.79}
& {\small 0.99} & {\small 0.9} & {\small 0.78} \\ 
${\small Orac\_sel}$ & {\small 0.78} & {\small 0} & \multicolumn{1}{c|}%
{\small 0} & {\small 0.51} & {\small 0} & {\small 0.04} & {\small 0.14} & 
{\small 0} & {\small 0.02} & {\small 0.15} \\ 
${\small Num\_irrel}$ & {\small 0.22} & {\small 1.12} & \multicolumn{1}{c|}%
{\small 2.11} & {\small 0.73} & {\small 21.18} & {\small 8} & {\small 3.67}
& {\small 20.38} & {\small 7.4} & {\small 3.37} \\ 
${\small in\_Score}$ & {\small 0.943} & {\small 0.965} & \multicolumn{1}{c|}%
{\small 0.972} & {\small 0.957} & {\small 0.981} & {\small 0.939} & {\small %
0.911} & {\small 0.977} & {\small 0.934} & {\small 0.907} \\ 
${\small in\_RS}$ & {\small 1.032} & {\small 1.056} & \multicolumn{1}{c|}%
{\small 1.063} & {\small 1.047} & {\small 1.073} & {\small 1.027} & {\small %
0.997} & {\small 1.070} & {\small 1.022} & {\small 0.993} \\ 
${\small out\_Score}$ & {\small 0.893} & {\small 0.891} & \multicolumn{1}{c|}%
{\small 0.883} & {\small 0.895} & {\small 0.876} & {\small 0.884} & {\small %
0.881} & {\small 0.876} & {\small 0.884} & {\small 0.880} \\ 
${\small out\_RS}$ & {\small 0.971} & {\small 0.969} & \multicolumn{1}{c|}%
{\small 0.960} & {\small 0.973} & {\small 0.953} & {\small 0.961} & {\small %
0.958} & {\small 0.953} & {\small 0.961} & {\small 0.957} \\ \hline
\end{tabular}%
\end{center}
\end{table}

To save space, we present details of the simulation results under the setup
of DGP(ii) in Online Appendix \ref{sim:hetero:design}. The results are
similar to those under DGP(i).

\section{An Illustrative Application\label{empirical application}}

We illustrate usefulness of PRESCIENCE in the empirical application of
work-trip mode choice. We used the transportation mode dataset analyzed by %
\citet{horowitz1993}. This dataset has also been well studied for
illustration of econometric methods developed for binary response models
(e.g., see \citet{Florios:Skouras:08}, \citet[Section 4.3]{benoit2012}, and
the references therein). 
The previous literature focused on estimating slope coefficients in the
binary response model; however, in this section, we are mainly interested in
the numerical performance of alternative MIO algorithms and the result of
covariate selection.

The data consist of 842 observations sampled randomly from the Washington,
D.C., area transportation study. Each record in the dataset contains the
following information for a single work trip of the traveler: the chosen
transportation mode, the number of cars owned by the traveler's household ($%
CARS$), the transit out-of-vehicle travel time minus automobile
out-of-vehicle travel time in minutes ($DOVTT$), the transit in-vehicle
travel time minus automobile in-vehicle travel time in minutes ($DIVTT$) and
the transit fare minus automobile travel cost in dollars ($DCOST$).

The dependent variable $Y$ is the traveler's chosen mode of transportation
such that $Y=1$ if the choice is automobile and 0 otherwise. Following %
\citet{Florios:Skouras:08}, we standardized each of explanatory variables to
have mean zero and unit variance. Following \citet{horowitz1993} and %
\citet{Florios:Skouras:08}, we specified the coefficient of $DCOST$ to be
unity and did not estimate that parameter. We set the focus covariates $%
X=(X_{0},\widetilde{X})$ to be $\left( DCOST,1\right) $, where the constant
term was included to capture the regression intercept and the parameter $%
\alpha $ was set to be unity. The resulting PRESCIENCE problem hence reduced
to the maximization problem (\ref{sub-problem}). We implemented the two MIO
formulations developed in Section \ref{implementation} for solving this
problem. To compare their computational performance, we report the CPU time
(in seconds) and the number of branch-and-bound nodes that the MIO solver
had explored to reach the optimal solution. The former depends on both the
computing hardware and software configurations whereas the latter only
depends on the solution algorithms employed by the MIO solver.

For the auxiliary covariates, we set%
\begin{align}
Z& =(CARS,DOVTT,DIVTT,CARS\times DOVTT,DOVTT\times DIVTT,  \notag \\
& CARS\times DIVTT,CARS\times CARS,DOVTT\times DOVTT,  \notag \\
& DIVTT\times DIVTT).  \label{quadratic expansion specification}
\end{align}%
The covariate specification (\ref{quadratic expansion specification}) allows
us to approximate a smooth function of the variables $(CARS,DOVTT,DIVTT)$ by
its quadratic expansion.\footnote{%
In Online Appendix \ref{emp:linear:spec}, we report empirical results using $%
Z=\left( CARS,DOVTT,DIVTT\right) $.} We are interested in the data driven
selection of these expansion terms through the PRESCIENCE procedure. In this
setup, we have that $k=1$ and $p=9$. We specified the parameter space $%
\Theta $ to be $\left[ -10,10\right] ^{10}$. We computed the PRESCIENCE for
each $q\in \{1,3,5\}$. Since there are 10 unknown parameters in this setup,
we solved for the PRESCIENCE solutions with a non-zero tolerance level which
was also specified according to the rule (\ref{early stopping rule}). For $%
n=842$, this amounts to setting the MIO tolerance level to be about 4.4\%.
To further reduce the computational cost, we adopted the warm-start strategy
in the resulting MIO formulations. We set $\tau =1.5$ and constructed $(%
\widehat{P}_{i})_{i=1}^{n}$ using the fitted choice probabilities from the
logit regression of $Y$ on all the covariates to derive the refined space $%
\widehat{\Theta }\left( \tau \right) $ from the initial parameter space $%
\Theta $. 
\begin{table}[tbh]
\caption{Refined parameter bounds ($\protect\tau =1.5$)}
\label{tab-04}
\begin{center}
\emph{Covariate specification}: $k=1,p=9$ \\[0pt]
\vspace*{2ex} 
\begin{tabular}{lcc}
\hline\hline
Variable & lower bound & upper bound \\ \hline
Intercept & -10 & 10 \\ 
$CARS$ & -10 & 10 \\ 
$DOVTT$ & -9.8299 & 9.8299 \\ 
$DIVTT$ & -8.0158 & 8.0158 \\ 
$CARS\times DOVTT$ & -6.0306 & 6.0306 \\ 
$DOVTT\times DIVTT$ & -7.5870 & 7.5870 \\ 
$CARS\times DIVTT$ & -5.7873 & 5.7873 \\ 
$CARS\times CARS$ & -4.2513 & 4.2513 \\ 
$DOVTT\times DOVTT$ & -1.9552 & 1.9552 \\ 
$DIVTT\times DIVTT$ & -5.7297 & 5.7297 \\ \hline
\end{tabular}%
\end{center}
\end{table}

Table 4 presents the refined parameter bounds derived from $\widehat{\Theta }%
\left( \tau \right) $ in this setup. The results of Table \ref{tab-04}
indicate that the size (measured in the volume of a ($k+p$) dimensional
cube) of $\widehat{\Theta }\left( \tau \right) $ is only about 0.99\% of
that of $\Theta $; thus, there is a considerable reduction in the parameter
search space by using the refined space $\widehat{\Theta }\left( \tau
\right) $ in place of the original space $\Theta $. 
\begin{table}[tbh]
\caption{Results Using Quadratic Expansion}
\label{tab-05}
\begin{center}
\emph{Covariate specification}: $k=1,p=9$ \\[0pt]
\vspace*{2ex} 
\begin{tabular}{|l|l|l|l|l|l|l|}
\hline\hline
$q$ & \multicolumn{2}{|l}{$1$} & \multicolumn{2}{|l}{$3$} & 
\multicolumn{2}{|l|}{$5$} \\ \hline
MIO formulation & (\ref{constrained MIO}) & (\ref{Florios-Skouras MIO}) & (%
\ref{constrained MIO}) & (\ref{Florios-Skouras MIO}) & (\ref{constrained MIO}%
) & (\ref{Florios-Skouras MIO}) \\ \hline
\multicolumn{7}{|l|}{focus covariates} \\ \hline
$DCOST$ & 1 & 1 & 1 & 1 & 1 & 1 \\ \hline
Intercept & 3.6374 & 3.2803 & 4.5191 & 3.1493 & 5.4587 & 2.9101 \\ \hline
\multicolumn{7}{|l|}{auxiliary covariates} \\ \hline
$CARS$ & 3.0404 & 2.4667 & 3.5782 & 2.4416 & 6.3592 & 2.1480 \\ \hline
$DOVTT$ & \multicolumn{1}{|l|}{0} & \multicolumn{1}{|l|}{0} & 0.8390 & 0.3416
& 0 & 0 \\ \hline
$DIVTT$ & \multicolumn{1}{|l|}{0} & \multicolumn{1}{|l|}{0} & 0 & 0 & 0 & 0
\\ \hline
$CARS\times DOVTT$ & \multicolumn{1}{|l|}{0} & \multicolumn{1}{|l|}{0} & 0 & 
0 & -1.1798 & -0.5332 \\ \hline
$DOVTT\times DIVTT$ & \multicolumn{1}{|l|}{0} & \multicolumn{1}{|l|}{0} & 0
& 0 & -3.7056 & -0.4177 \\ \hline
$CARS\times DIVTT$ & \multicolumn{1}{|l|}{0} & \multicolumn{1}{|l|}{0} & 
-0.2744 & 0.1644 & 0 & 0 \\ \hline
$CARS\times CARS$ & \multicolumn{1}{|l|}{0} & \multicolumn{1}{|l|}{0} & 0 & 0
& 0 & 0 \\ \hline
$DOVTT\times DOVTT$ & \multicolumn{1}{|l|}{0} & \multicolumn{1}{|l|}{0} & 0
& 0 & 1.3282 & 0.1835 \\ \hline
$DIVTT\times DIVTT$ & \multicolumn{1}{|l|}{0} & \multicolumn{1}{|l|}{0} & 0
& 0 & 2.7936 & 0.0744 \\ \hline
\multicolumn{7}{|l|}{in-sample performance} \\ \hline
maximized average score & 0.8979 & 0.8979 & 0.9086 & 0.9086 & 0.9145 & 0.9097
\\ \hline
\multicolumn{7}{|l|}{MIO solver output} \\ \hline
$MIO\_gap$ & 0.0428 & 0.0428 & 0.0428 & 0.0428 & 0.0428 & 0.0428 \\ \hline
CPU time (in seconds) & 55 & 259 & 65985 & 3931 & 566 & 778 \\ \hline
branch-and-bound nodes & 24547 & 71943 & 1521685 & 439394 & 208537 & 425053
\\ \hline
\end{tabular}%
\end{center}
\end{table}

We now present in Table \ref{tab-05} the estimation results for the setup
with the covariate specification (\ref{quadratic expansion specification}).
From Table \ref{tab-05}, we can see that the two MIO formulations (\ref%
{constrained MIO}) and (\ref{Florios-Skouras MIO}) yield the same set of
selected variables across all the three cases of $q$. The parameter
estimates computed from both formulations are also qualitatively similar in
general. The variable $CARS$ remains to be selected in all these cases and
its parameter estimate is also of the largest magnitude among all parameter
estimates of the quadratic expansion variables. Moreover, there is very
little loss in the goodness of fit from adopting only $CARS$ as the
auxiliary covariate.

We now remark on the computational performance. For the cases of $q\in
\{1,5\}$, formulation (\ref{constrained MIO}) clearly outperformed
formulation (\ref{Florios-Skouras MIO}) in both the CPU time and the number
of branch-and-bounds used in the computation; however, both MIO formulations
performed quite well for these two variable selection cases. By contrast,
for the case of $q=3$, both approaches incurred more computational cost. In
particular, it took around 18.3 hours and noticeably more branch-and-bound
nodes to solve the formulation (\ref{constrained MIO}) in the $q=3$
scenario. On the whole, these results suggest that the two MIO formulations (%
\ref{constrained MIO}) and (\ref{Florios-Skouras MIO}) are valuable
complements for implementing the PRESCIENCE procedures.

By construction, the maximized average score of the PRESCIENCE approach
increases with the specified value of $q$. In this empirical application,
our results presented so far indicate that the in-sample predictive
performance of the parsimonious predictive model using an intercept term and
only the two variables $\left( DCOST,CARS\right) $ seems comparable to that
of a more complex model using a richer set of covariates. We now investigate
this issue further via the method of cross validation (CV).

We conducted the 5-fold CV analysis to assess the out-of-sample predictive
performance of the PRESCIENCE methods for the selection of the variables
specified by (\ref{quadratic expansion specification}). As in Table \ref%
{tab-05}, we considered the cases of $q\in \{1,3,5\}$ for which we
implemented the corresponding PRESCIENCE procedures with a non-zero
tolerance level. Because the training sample in each CV fold contains around
80\% of the original observations, we set the MIO tolerance level for early
termination to be about 4.9\% by the rule (\ref{early stopping rule}).

\begin{table}[tbh]
\caption{Summary of the 5-fold Cross Validation Results}
\label{tab-06}
\begin{center}
\emph{Covariate specification}: $k=1,p=9$ \\[0pt]
\vspace*{2ex} 
\begin{tabular}{|l|l|l|l|l|l|l|}
\hline\hline
$q$ & \multicolumn{2}{|l}{$1$} & \multicolumn{2}{|l}{$3$} & 
\multicolumn{2}{|l|}{$5$} \\ \hline
MIO formulation & (\ref{constrained MIO}) & (\ref{Florios-Skouras MIO}) & (%
\ref{constrained MIO}) & (\ref{Florios-Skouras MIO}) & (\ref{constrained MIO}%
) & (\ref{Florios-Skouras MIO}) \\ \hline
\multicolumn{7}{|l|}{average in-sample performance} \\ \hline
$MIO\_gap$ & 0.0454 & 0.0475 & 0.0475 & 0.0475 & 0.0475 & 0.0475 \\ \hline
maximized objective value & 0.8993 & 0.8993 & 0.9100 & 0.9112 & 0.9127 & 
0.9148 \\ \hline
\multicolumn{7}{|l|}{average out-of-sample performance} \\ \hline
proportion of correct predictions & 0.9026 & 0.8979 & 0.8884 & 0.8812 & 
0.8932 & 0.8884 \\ \hline
\end{tabular}%
\end{center}
\end{table}

Table \ref{tab-06} summarizes the 5-fold CV results, which are based on the
averages over the performance results computed in each CV fold. From Table %
\ref{tab-06}, we can see that, for both MIO formulations, the in-sample
maximized objective values were already very similar across the three cases
of $q$ though they did strictly increase with $q$. Moreover, irrespective of
the MIO formulations, the parsimonious case of $q=1$ had the best
out-of-sample performance.

\begin{figure}[h!tb]
\caption{Auxiliary Covariate Parameter Estimates in each CV-fold}
\label{fig-01}
\begin{center}
\emph{Covariate specification}: $k=1,p=9$ \\[0pt]
\vspace*{2ex} \makebox{
\includegraphics[scale=.61]{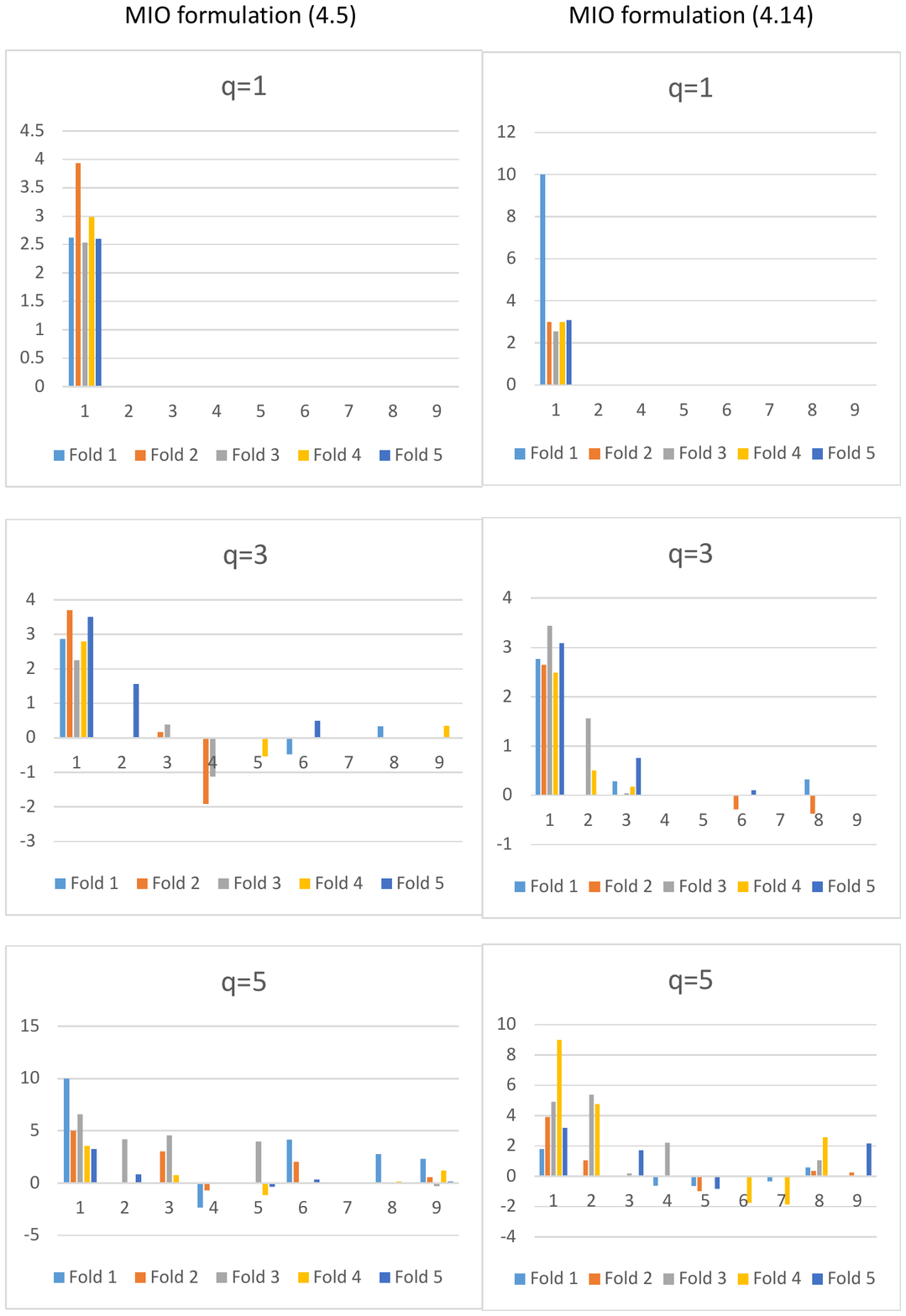}
}
\end{center}
\par
\parbox{6in}{
Notes: For each panel of Figure \ref{fig-01}, the values in the horizontal axis correspond to
the indices of the 9 components of the auxiliary covariate vector defined by
(\ref{quadratic expansion specification}), whereas the vertical axis
displays the values of the parameter estimates.
The left and right panels are based on MIO formulations \eqref{constrained MIO}
and \eqref{Florios-Skouras MIO}, respectively. 
}
\end{figure}

We now inspect the variables selected in each of the 5 CV folds. Figure \ref%
{fig-01} summarizes the parameter estimates computed in each CV fold for the
9 auxiliary covariates specified by (\ref{quadratic expansion specification}%
). From this figure, we also note that $CARS$ was selected across all $q$
cases in all CV folds. Its parameter estimate was also of a relatively large
magnitude when compared to those of other selected variables. These
cross-validation results further strengthen the finding that $CARS$ may be
the most important predictive variable for the work-trip mode choice.

\section{Concluding Remarks\label{conclusion}}

In this paper, we consider the variable selection problem for predicting
binary outcomes. We study the best subset selection procedure by which the
covariates are chosen by maximizing the maximum score objective function
subject to a constraint on the maximal number of selected covariates. We
establish non-asymptotic upper and lower risk bounds for the resulting best
subset maximum score binary prediction rule when the dimension of potential
covariates is possibly much larger than the sample size. The derived upper
and lower risk bounds are minimax rate-optimal when the maximal number of
selected variables is fixed and does not increase with the sample size. For
implementation, we show that the variable selection problem of this paper
can be equivalently formulated as a mixed integer optimization problem,
which enables computation of the exact or an approximate solution with a
definite approximation error bound.

The present paper takes the maximum score approach for the binary prediction
problem. There is a large body of the literature that studies maximum score
estimation in various other aspects since the seminal work by %
\citet{Manski:1975, Manski:1985}. In the context of semiparametric binary
response models, advances of the maximum score approach have been made in
terms of point identification \citep{Manski:1988}, partial identification 
\citep{Manski2002, komarova2013, blevins2015,
chen2015}, asymptotic distribution \citep{kim1990,Seo2016}, panel data %
\citep{Manski87, CMvS:95, abrevaya2000}, time series 
\citep{Moon2004,
Guerre2006, deJong2011}, dynamic network formation \citep{Graham:2016},
nonparametrically generated regressors \citep{Chen-et-al:14}, and so on. The
numerical approach employed in this paper can be adapted to these contexts.

Efficient computation is particularly demanding when it is necessary to
obtain maximum score estimates repeatedly many times. This difficulty
naturally arises, for example, in the context of resampling %
\citep{delgado2001, abrevaya2005, Lee:Pun:06, Patra2015} and change-point
problems \citep{Lee2008, Lee-et-al:11}. It would be an interesting future
research topic to investigate numerical performance of our method (e.g. the
warm-start procedure in Section \ref{sec:para-space}) in these
computation-intensive problems.

The maximum score approach has produced many offsprings: smoothed maximum
score estimation \citep{Horowitz:92,  horowitz2002}, multinomial choice
estimation \citep{Matzkin:93, Fox2007}, integrated maximum score estimation %
\citep{Chen2010}, Bayesian method \citep{benoit2012}, alternative estimation
based on local nonlinear least squares 
\citep{blevins-khan-2013-dfbr,
blevins-khan-2013, Khan2013}, estimation using local polynomial smoothing %
\citep{Chen:Zhang:2015}, and non-Bayesian Laplace type estimation %
\citep{JPW15, JPW14} among many others. Some of these alternative estimation
methods are equipped with algorithms that are easier to compute than the
maximum score estimation. It is an interesting open question how to
accommodate a variable selection problem in these alternative methods. One
starting point can be the work by \citet{Jiang:10} who considered empirical
risk minimization under the $\ell _{0}$ constraint with Gibbs posterior as
an alternative to maximum score prediction. It might be also interesting to
consider penalized estimation with the $\ell _{0}$, $\ell _{1}$ and/or $\ell
_{2}$ penalty. The maximum score approach is also closely related to the
maximum utility estimation framework 
\citep{LieliNieto-Barthaburu, LieliWhite, Elliott2013,
LieliSpringborn} for binary decision under model uncertainty. It would be
interesting to generalize our results to this framework. These are also
topics for future research.


\appendix%


\section*{Appendices}

Appendix \ref{sec:appendix} presents proofs of all
theoretical results, 
Appendix  \ref{bb-method} gives a brief description of the 
 branch-and-bound method for solving the mixed integer optimization
problem,
Appendix \ref{sim:hetero:design} reports
the simulation results under the
setup of DGP(ii) of Section \ref{simulation study}, 
Appendix \ref{simulations of the cold and warm start methods} reports an additional 
simulation study on the performance of adopting the warm-start strategy of
Section \ref{sec:para-space} in solving the MIO formulations (\ref%
{constrained MIO}) and (\ref{Florios-Skouras MIO}),
and 
Appendix \ref{emp:linear:spec} reports additional results for the empirical example using the
linear specification of covariates.

\section{Proofs of theoretical results}

\label{sec:appendix}

\subsection{Proofs of Propositions \protect\ref{observational equivalence}, 
\protect\ref{margin condition for the binary response model}, and \protect
\ref{primitive assumptions on norm equivalence}}

\begin{proof}[Proof of Proposition \protect\ref{observational equivalence}]
By (\ref{binary response model}) and (\ref{median independence}), we have
that%
\begin{equation}
w^{\prime }\theta ^{\ast }\geq 0\Longleftrightarrow \eta (w)\geq
1/2\Longleftrightarrow b^{\ast }(w)=1.  \label{r1}
\end{equation}%
Since the functional $\widetilde{S}(b^{\ast })$ is maximized at $b=b^{\ast }$
and, by Condition \ref{margin condition for upper bound}, this maximum is
unique, we have that%
\begin{equation}
b^{\ast }\in \mathcal{B}_{q}\Longleftrightarrow b^{\ast }=b_{\theta _{0}}.
\label{r2}
\end{equation}%
By (\ref{predictor}), we have that%
\begin{equation}
b_{\theta _{0}}(w)=1\Longleftrightarrow w^{\prime }\theta _{0}\geq 0.
\label{r3}
\end{equation}%
Thus it follows from (\ref{r1}), (\ref{r2}) and (\ref{r3}) that $b^{\ast
}\in \mathcal{B}_{q}$ if and only if $W^{\prime }\theta _{0}$ and $W^{\prime
}\theta ^{\ast }$ have the same sign with probability 1.
\end{proof}

\begin{proof}[Proof of Proposition \protect\ref{margin condition for the
binary response model}]
Let $G_{w}(t)\equiv P\left( \xi \leq t|W=w\right) $. By (\ref{binary
response model}) and (\ref{median independence}), we have that $\eta
(w)=G_{w}(w^{\prime }\theta ^{\ast })$ and $G_{w}(0)=1/2$. By condition (ii)
of Proposition \ref{margin condition for the binary response model}, it
hence follows that, if $w^{\prime }\theta ^{\ast }\geq \kappa _{1}$, then 
\begin{equation*}
\eta (w)\geq G_{w}(\kappa _{1})\geq 1/2+\kappa _{1}\kappa _{2}.
\end{equation*}

On the other hand, by similar arguments, we have that, if $w^{\prime }\theta
^{\ast }\leq -\kappa _{1}$, then%
\begin{equation*}
\eta (w)\leq G_{w}(-\kappa _{1})\leq 1/2-\kappa _{1}\kappa _{2}.
\end{equation*}%
Combing these results and using condition (i) of Proposition \ref{margin
condition for the binary response model}, we thus deduce that condition (\ref%
{margin condition}) holds with $h=2\kappa _{1}\kappa _{2}$.
\end{proof}

\begin{proof}[Proof of Proposition \protect\ref{primitive assumptions on
norm equivalence}]
Consider any two vectors $\theta =\left( \alpha ,\beta ,\gamma \right) ,$ $%
\widetilde{\theta }=(\widetilde{\alpha },\widetilde{\beta },\widetilde{%
\gamma })\in \{-1,1\}\times \Theta _{q}$ such that $\alpha =\widetilde{%
\alpha }$ and $\beta =\widetilde{\beta }$. We now prove this proposition for
the case $\alpha =\widetilde{\alpha }=1$. The other case $\alpha =\widetilde{%
\alpha }=-1$ can be proved using identical arguments and hence is omitted.

Assume $\alpha =\widetilde{\alpha }=1.$ Note that 
\begin{eqnarray*}
&&\left\Vert b_{\theta }-b_{\widetilde{\theta }}\right\Vert _{1} \\
&=&P\left( -\widetilde{X}\beta -Z^{\prime }\gamma \leq X_{0}<-\widetilde{X}%
\widetilde{\beta }-Z^{\prime }\widetilde{\gamma }\right) +P\left( -%
\widetilde{X}\beta -Z^{\prime }\gamma >X_{0}\geq -\widetilde{X}\widetilde{%
\beta }-Z^{\prime }\widetilde{\gamma }\right) \\
&=&E\left[ P\left( -\widetilde{X}\beta -Z^{\prime }\gamma \leq X_{0}<-%
\widetilde{X}\widetilde{\beta }-Z^{\prime }\widetilde{\gamma }|\widetilde{X}%
,Z\right) 1\left\{ Z^{\prime }\gamma \geq Z^{\prime }\widetilde{\gamma }%
\right\} \right] \\
&&+E\left[ P\left( -\widetilde{X}\beta -Z^{\prime }\gamma >X_{0}\geq -%
\widetilde{X}\widetilde{\beta }-Z^{\prime }\widetilde{\gamma }|\widetilde{X}%
,Z\right) 1\left\{ Z^{\prime }\gamma \leq Z^{\prime }\widetilde{\gamma }%
\right\} \right] .
\end{eqnarray*}%
By Condition (a) of Proposition \ref{primitive assumptions on norm
equivalence}, we hence have that 
\begin{equation}
E\left[ \left\vert Z^{\prime }\left( \gamma -\widetilde{\gamma }\right)
\right\vert \right] L_{1}^{-1}\leq \left\Vert b_{\theta }-b_{\widetilde{%
\theta }}\right\Vert _{1}\leq E\left[ \left\vert Z^{\prime }\left( \gamma -%
\widetilde{\gamma }\right) \right\vert \right] L_{1}.  \label{b1}
\end{equation}%
Let $J\equiv \{j\in \left\{ 1,...,p\right\} :\gamma _{j}\neq \widetilde{%
\gamma }_{j}\}$. Since $\left\Vert \gamma \right\Vert _{0}\leq q$ and $%
\left\Vert \widetilde{\gamma }\right\Vert _{0}\leq q$, we have that $%
\left\vert J\right\vert \leq 2q$. Therefore, $Z^{\prime }\left( \gamma -%
\widetilde{\gamma }\right) =Z_{J}^{\prime }\delta _{J}$ where $Z_{J}$
denotes the subvector of $Z\equiv (Z^{(1)},\ldots ,Z^{(p)})^{\prime }$
formed by keeping only those elements $Z^{(j)}$ with $j\in J$ and $\delta
_{J}$ denotes the subvector of $\gamma -\widetilde{\gamma }$ formed by
keeping only those elements $(\gamma _{j}-\widetilde{\gamma }_{j})$ with $%
j\in J$.

By Condition (b) and Cauchy-Schwarz inequality, we have that with
probability 1,%
\begin{equation}
\left\vert Z_{J}^{\prime }\delta _{J}\right\vert \leq L_{2}\left\Vert \delta
_{J}\right\Vert _{E}  \label{b2}
\end{equation}%
and hence%
\begin{equation}
\frac{\delta _{J}^{\prime }Z_{J}Z_{J}^{\prime }\delta _{J}}{%
L_{2}^{2}\left\Vert \delta _{J}\right\Vert _{E}^{2}}\leq \frac{\left\vert
Z_{J}^{\prime }\delta _{J}\right\vert }{L_{2}\left\Vert \delta
_{J}\right\Vert _{E}}\leq 1.  \label{b3}
\end{equation}%
Using (\ref{b3}) and the assumption that the smallest eigenvalue of $E\left(
Z_{J}Z_{J}^{\prime }\right) $ is bounded below by $L_{3}$, we thus have that 
\begin{equation}
E\left[ \left\vert Z_{J}^{\prime }\delta _{J}\right\vert \right] \geq
L_{2}^{-1}L_{3}\left\Vert \delta _{J}\right\Vert _{E}.  \label{b4}
\end{equation}%
Noting that $\left\Vert \delta _{J}\right\Vert _{E}=\left\Vert \theta -%
\widetilde{\theta }\right\Vert _{E}$ and combining (\ref{b1}), (\ref{b2})
and (\ref{b4}), we conclude that condition (\ref{condition on the
distribution of W}) holds with $c_{u}=L_{1}L_{2}$ and $%
c_{l}=(L_{1}L_{2})^{-1}L_{3}$.
\end{proof}

\subsection{Proof of Theorem \protect\ref{maximal inequality}}

Recall the notation $\theta=(\alpha,\beta,\gamma)$. Let 
\begin{equation*}
G_{n}(\theta)\equiv\sqrt{n}\left[ S_{n}(\theta)-S(\theta)\right] .
\end{equation*}
We first present the following lemma, which will be used to prove Theorem %
\ref{maximal inequality}.

\begin{lemma}
\label{tail probability bound}For $t>0$, there is a universal constant $D$
such that
\end{lemma}

\begin{equation*}
P\left( \sup\limits_{\theta\in\{-1,1\}\times\Theta_{q}}\left\vert
G_{n}(\theta)\right\vert >t\right) \leq2\binom{p}{q}\left( \frac{Dt}{\sqrt{%
4s+4}}\right) ^{4s+4}e^{-2t^{2}}.
\end{equation*}

\begin{proof}
Let $m$ be a subset of the index set $\left\{ 1,...,p\right\} $ such that $m$
contains only $q$ elements. Let $\mathcal{M}$ be the collection of all such
subsets. Note that $\left\vert \mathcal{M}\right\vert =\binom{p}{q}$. For $%
m\in\mathcal{M}$, let%
\begin{equation*}
\Gamma_{m}\equiv\{\left( \beta,\gamma\right) \in\Theta:\gamma_{j}=0\text{
for }j\notin m\}.
\end{equation*}

For any $t>0$, we have that%
\begin{align}
P\left( \sup\limits_{\theta\in\{-1,1\}\times\Theta_{q}}\left\vert
G_{n}(\theta)\right\vert >t\right) & \leq\sum_{m\in\mathcal{M}}P\left(
\sup\limits_{\left( \beta,\gamma\right) \in\Gamma_{m}}\left\vert
G_{n}(1,\beta,\gamma)\right\vert >t\right)  \notag \\
& +\sum_{m\in\mathcal{M}}P\left( \sup\limits_{\left( \beta,\gamma\right)
\in\Gamma_{m}}\left\vert G_{n}(-1,\beta,\gamma)\right\vert >t\right) .
\label{ineq 1}
\end{align}
To complete the proof, it remains to derive the bounds on the tail
probability terms on the right hand side of the inequality above.

Consider the function $f_{\theta }:\{0,1\}\times \mathcal{W\mapsto }\{0,1\}$
defined by%
\begin{align*}
f_{\theta }(y,w)& \equiv 1\left\{ y=b_{\theta }(w)\right\} \\
& =1-y+\left( 2y-1\right) b_{\theta }.
\end{align*}%
For $m\in \mathcal{M}$, let 
\begin{align*}
\digamma _{m}^{+}& \equiv \{f_{\theta }:\theta \in \{1\}\times \Gamma _{m}\},
\\
\digamma _{m}^{-}& \equiv \{f_{\theta }:\theta \in \{-1\}\times \Gamma
_{m}\}.
\end{align*}%
For each $m\in \mathcal{M}$, by Lemmas 9.6 and 9.9 of \citet{Kosorok2008},
the class of functions $\digamma _{m}^{+}$ \ and $\digamma _{m}^{-}$ are
both $VC$ classes of functions with $VC$ indices $V(\digamma _{m}^{+})$ and $%
V(\digamma _{m}^{-})$ satisfying that 
\begin{equation*}
V(\digamma _{m}^{+})=V(\digamma _{m}^{-})\leq 2s+3.
\end{equation*}%
For two measurable functions $f$ and $g$ and a given probability measure $Q$%
, define the (semi-) metric 
\begin{equation*}
d_{Q}(f,g)\equiv \sqrt{\int \left( f-g\right) ^{2}dQ}.
\end{equation*}%
By Theorem 9.3 of \citet{Kosorok2008}, we have that, for some universal
constant $K>1$ and $0<\epsilon <1$,%
\begin{align*}
& \sup\nolimits_{Q}N\left( \epsilon ,\digamma _{m}^{+},d_{Q}\right) \vee
\sup\nolimits_{Q}N\left( \epsilon ,\digamma _{m}^{-},d_{Q}\right) \\
& \leq K\left( 2s+3\right) \left( 16e\right) ^{2s+3}\left( \epsilon \right)
^{-4(s+1)} \\
& \leq \left( \frac{\Lambda (s)}{\epsilon }\right) ^{4s+4},
\end{align*}%
where%
\begin{equation*}
\Lambda (s)\equiv \left[ K\left( 2s+3\right) \left( 16e\right) ^{2s+3}\right]
^{\frac{1}{4s+4}}
\end{equation*}%
and, for a given class of functions $\digamma $, $N\left( \epsilon ,\digamma
,d_{Q}\right) $ denotes the minimal number of open balls (defined under the
metric $d_{Q}$) of radius $\epsilon $ required to cover $\digamma $.

Observe that 
\begin{equation}
\Lambda(s)\leq16Ke\left( 2s+3\right) ^{\frac{1}{4s+4}}\leq 32Ke,
\label{uniform bound on Gamma(s)}
\end{equation}
where the second inequality follows from the fact that $i\leq2^{i}$ for all
integer $i$.

Using (\ref{uniform bound on Gamma(s)}), we can apply Theorem 1.3 of %
\citet{Talagrand1994} to deduce that, for $t>0$, there is a universal
constant $D$ such that%
\begin{align}
& P\left( \sup\limits_{\left( \beta ,\gamma \right) \in \Gamma
_{m}}\left\vert G_{n}(1,\beta ,\gamma )\right\vert >t\right) \vee P\left(
\sup\limits_{\left( \beta ,\gamma \right) \in \Gamma _{m}}\left\vert
G_{n}(-1,\beta ,\gamma )\right\vert >t\right)  \notag \\
& \leq \left( \frac{Dt}{\sqrt{4s+4}}\right) ^{4s+4}e^{-2t^{2}}.
\label{Talagrand result}
\end{align}%
Lemma \ref{tail probability bound} hence follows by combining inequalities (%
\ref{ineq 1}) and (\ref{Talagrand result}).
\end{proof}

We now prove Theorem \ref{maximal inequality}.

\begin{proof}[Proof of Theorem \protect\ref{maximal inequality}]
By Lemma \ref{tail probability bound} and using the fact that $\binom{p}{q}%
\leq p^{q},$ we have that, for some universal constant $D$ and for $t>0$,%
\begin{equation}
P\left( \sup\limits_{\theta\in\{-1,1\}\times\Theta_{q}}\left\vert
G_{n}(\theta)\right\vert >t\right) \leq2p^{q}\left( \frac{Dt}{\sqrt{4s+4}}%
\right) ^{4s+4}e^{-2t^{2}}.  \label{ineq 2}
\end{equation}
For $t\geq D$ and $s\geq1$, the right hand side term of (\ref{ineq 2}) can
be further bounded above by%
\begin{equation}
p^{s}t^{8s+8}2^{-6s-5}e^{-2t^{2}}=e^{\lambda(s,p,t)},  \label{lamda}
\end{equation}
where%
\begin{equation*}
\lambda(s,p,t)\equiv-2t^{2}+\left( 8s+8\right) \ln t+s\ln p-\left(
6s+5\right) \ln2.
\end{equation*}
For $\sigma>0$, let 
\begin{equation}
t=\sqrt{M_{\sigma}r_{n}}  \label{t}
\end{equation}
where%
\begin{equation}
M_{\sigma}\equiv\left( 1+\sigma/2\right) \vee D^{2}.  \label{M_eta bound}
\end{equation}
Note that $t\geq D$ by (\ref{t}) and (\ref{M_eta bound}). Thus with the $t$
value specified by (\ref{t}), we have that 
\begin{equation}
\lambda(s,p,t)\leq\left( -2M_{\sigma}+1\right) r_{n}+\left( 4s+4\right)
\ln\left( M_{\sigma}r_{n}\right) -\left( 6s+5\right) \ln2\leq-\sigma r_{n}
\label{lamda bound}
\end{equation}
where the second inequality follows from (\ref{M_eta bound}) and (\ref%
{condition on r_n}).

It is straightforward to see that%
\begin{align*}
U_{n}& =\sup\limits_{\theta \in \{-1,1\}\times \Theta _{q}}\text{ }S(\theta
)-S(\widehat{\theta }) \\
& \leq \sup\limits_{\theta \in \{-1,1\}\times \Theta _{q}}\left\vert
S_{n}(\theta )-S(\theta )\right\vert +\sup\limits_{\theta \in \{-1,1\}\times
\Theta _{q}}S_{n}(\theta )-S(\widehat{\theta }) \\
& \leq \sup\limits_{\theta \in \{-1,1\}\times \Theta _{q}}\left\vert
S_{n}(\theta )-S(\theta )\right\vert +S_{n}(\widehat{\theta })-S(\widehat{%
\theta }) \\
& \leq 2\sup\limits_{\theta \in \{-1,1\}\times \Theta _{q}}\left\vert
S_{n}(\theta )-S(\theta )\right\vert .
\end{align*}%
Hence, we have that%
\begin{equation}
P\left( U_{n}>2\sqrt{\frac{M_{\sigma }r_{n}}{n}}\right) \leq P\left(
\sup\limits_{\theta \in \{-1,1\}\times \Theta _{q}}\left\vert G_{n}(\theta
)\right\vert >\sqrt{M_{\sigma }r_{n}}\right) .  \label{uniform bound}
\end{equation}%
Therefore, Theorem \ref{maximal inequality} follows by putting together the
results of (\ref{ineq 2}), (\ref{lamda}) and (\ref{lamda bound}) and then
concluding that the right hand side term of (\ref{uniform bound}) is bounded
above by $e^{-\sigma r_{n}}$ under condition (\ref{condition on r_n}).
\end{proof}

\subsection{Proof of Theorem \protect\ref{risk upper bound}}

We first introduce some notation which will be used in the proof Theorem \ref%
{risk upper bound}. Let $\mathcal{A}$ be a collection of subsets of $%
\mathcal{W}$. For any subset $S\subset\mathcal{W}$, let $T_{\mathcal{A}}(S)$
denote the trace of $\mathcal{A}$ on $S$ defined by%
\begin{equation*}
T_{\mathcal{A}}(S)=\left\{ A\cap S:A\in\mathcal{A}\right\} .
\end{equation*}
If $\mathcal{A=\cup}_{j\in J}\mathcal{A}_{j}$, then we have 
\begin{equation}
T_{\mathcal{A}}(S)\subset\mathcal{\cup}_{j\in J}T_{\mathcal{A}_{j}}(S)\text{
for }S\subset\mathcal{W}\text{.}  \label{trace}
\end{equation}

We now prove Theorem \ref{risk upper bound}.

\begin{proof}[Proof of Theorem \protect\ref{risk upper bound}]
By (\ref{predictor}), (\ref{S(a,b,c)}), (\ref{theoretical prediction rule})
and (\ref{S_tilda}), we have that%
\begin{equation}
S(\widehat{\theta })=\widetilde{S}\left( b_{\widehat{\theta }}\right) \text{
and }S_{q}^{\ast }=\sup\nolimits_{b\in \mathcal{B}_{q}}\widetilde{S}\left(
b\right) .  \label{S and S_tilda}
\end{equation}%
Theorem \ref{risk upper bound} is an application of Theorem 2 of %
\citet{massart2006} to the binary prediction problem. Massart and N\'{e}d%
\'{e}lec\thinspace (2006, Section 2.4) showed how to apply their Theorem 2
to derive the risk upper bound for the empirical risk minimizer. Using their
derived results (Massart and N\'{e}d\'{e}lec\thinspace (2006, p.\thinspace
2340)) in their Theorem 2 for our setup, we conclude that, under Condition %
\ref{margin condition for upper bound}, there are universal constants $K$
and $K^{\prime }$ such that%
\begin{equation}
\begin{split}
& E\left[ \widetilde{S}\left( b^{\ast }\right) -\widetilde{S}\left( b_{%
\widehat{\theta }}\right) \right] \\
& \leq 2\left[ \widetilde{S}\left( b^{\ast }\right) -\sup\nolimits_{b\in 
\mathcal{B}_{q}}\widetilde{S}\left( b\right) \right] \\
& +K^{\prime }\left[ \left( \frac{K^{2}\left( 1\vee E(H_{\mathcal{A}%
})\right) }{nh}\right) ^{\vartheta /\left( 2\vartheta -1\right) }\wedge 
\sqrt{\frac{K^{2}\left( 1\vee E(H_{\mathcal{A}})\right) }{n}}\right] ,
\end{split}
\label{upper bound due to Massart and Nedelec}
\end{equation}%
where 
\begin{align*}
\mathcal{A}& \mathcal{\equiv }\{A_{\theta }:\theta \in \{-1,1\}\times \Theta
_{q}\}, \\
A_{\theta }& \equiv \{w\in \mathcal{W}:w^{\prime }\theta \geq 0\},
\end{align*}%
and $H_{\mathcal{A}}$ is the random combinatorial entropy of $\mathcal{A}$
defined by%
\begin{equation*}
H_{\mathcal{A}}=\ln \left( \left\vert T_{\mathcal{A}}(\left\{
W_{1},W_{2},...,W_{n}\right\} )\right\vert \right) ,
\end{equation*}%
where 
\begin{equation}
T_{\mathcal{A}}(\left\{ W_{1},W_{2},...,W_{n}\right\} )=\left\{ A\cap
\left\{ W_{1},W_{2},...,W_{n}\right\} :A\in \mathcal{A}\right\}
\label{random trace}
\end{equation}%
is the trace of $\mathcal{A}$ on the covariate data $\left\{
W_{1},W_{2},...,W_{n}\right\} $.

Using (\ref{U}) and (\ref{S and S_tilda}), it follows that 
\begin{eqnarray}
U_{n} &=&\sup\nolimits_{b\in \mathcal{B}_{q}}\widetilde{S}\left( b\right) -%
\widetilde{S}\left( b_{\widehat{\theta }}\right)  \notag \\
&=&\left[ \sup\nolimits_{b\in \mathcal{B}_{q}}\widetilde{S}\left( b\right) -%
\widetilde{S}\left( b^{\ast }\right) \right] +\left[ \widetilde{S}\left(
b^{\ast }\right) -\widetilde{S}\left( b_{\widehat{\theta }}\right) \right] .
\label{Un}
\end{eqnarray}%
To complete the proof, it thus remains to derive an upper bound on the term $%
E(H_{\mathcal{A}})$.

Let $m$ be a subset of the index set $\left\{ 1,...,p\right\} $ such that $m$
contains only $q$ elements. Let $\mathcal{M}$ be the collection of all such
subsets. Note that $\left\vert \mathcal{M}\right\vert =\binom{p}{q}$. For $%
m\in \mathcal{M}$, let 
\begin{align*}
\mathcal{A}_{m}^{+}& \mathcal{\equiv }\{A_{\theta }:\theta \in \{1\}\times
\Gamma _{m}\}, \\
\mathcal{A}_{m}^{-}& \mathcal{\equiv }\{A_{\theta }:\theta \in \{-1\}\times
\Gamma _{m}\}, \\
\Gamma _{m}& \equiv \{\left( \beta ,\gamma \right) \in \Theta :\gamma _{j}=0%
\text{ for }j\notin m\}\text{.}
\end{align*}%
It is immediate to see that%
\begin{equation}
\mathcal{A=\cup }_{m\in \mathcal{M}}\left( \mathcal{A}_{m}^{+}\cup \mathcal{A%
}_{m}^{-}\right) .  \label{A}
\end{equation}%
For each $m\in \mathcal{M}$, by Lemmas 9.6 and 9.9 of \citet{Kosorok2008},
the family of sets $\mathcal{A}_{m}^{+}$ \ and $\mathcal{A}_{m}^{-}$ are
both $VC$ classes of sets with $VC$ indices $V(\mathcal{A}_{m}^{+})$ and $V(%
\mathcal{A}_{m}^{-})$ satisfying that 
\begin{equation*}
V(\mathcal{A}_{m}^{+})=V(\mathcal{A}_{m}^{-})\leq s+2.
\end{equation*}%
Hence by Corollary 1.3 of \citet{Lugosi:2002}, we have that%
\begin{equation}
\left\vert T_{\mathcal{A}_{m}^{+}}(\left\{ W_{1},W_{2},...,W_{n}\right\}
)\right\vert \vee \left\vert T_{\mathcal{A}_{m}^{-}}(\left\{
W_{1},W_{2},...,W_{n}\right\} )\right\vert \leq (n+1)^{s+1}.
\label{shattering coefficient bound}
\end{equation}%
By (\ref{trace}), (\ref{A}) and (\ref{shattering coefficient bound}), we
thus have that%
\begin{align}
H_{\mathcal{A}}& \leq \ln 2+\ln \binom{p}{q}+\left( s+1\right) \ln \left(
n+1\right)  \notag \\
& \leq \ln 2+q\ln p+\left( s+1\right) \ln \left( n+1\right) .
\label{bound on random combinatorial entropy}
\end{align}%
Theorem \ref{risk upper bound} therefore follows by combining the results (%
\ref{upper bound due to Massart and Nedelec}), (\ref{Un}) and (\ref{bound on
random combinatorial entropy}).
\end{proof}

\subsection{Proof of Theorem \protect\ref{risk lower bound}}

\begin{proof}[Proof of Theorem \protect\ref{risk lower bound}]
Define%
\begin{equation*}
R_{n}(h,\mathcal{B}_{q})=\inf_{\widehat{b}\in \mathcal{B}_{q}}\sup_{F\in 
\mathcal{P}(h,\mathcal{B}_{q})}E_{F}\left[ \widetilde{S}\left( b^{\ast
}\right) -\widetilde{S}(\widehat{b})\right]
\end{equation*}%
where the infimum is taken over the set of all binary predictors in $%
\mathcal{B}_{q}$ that are constructed based on the data $\left(
Y_{i},W_{i}\right) _{i=1}^{n}$. By (\ref{theoretical prediction rule}), it
follows that $S_{q}^{\ast }=\widetilde{S}\left( b^{\ast }\right) $ under $%
F\in \mathcal{P}(h,\mathcal{B}_{q})$. Thus Theorem \ref{risk lower bound} is
proved once we show that 
\begin{equation}
R_{n}(h,\mathcal{B}_{q})\geq \frac{\phi qc_{l}\left( 1-\phi \right) \left(
1-h\right) }{32nhc_{u}}\ln \left( \frac{p-q}{q/2}\right) .
\label{risk lower bound result}
\end{equation}

For any indicator function $b\in \mathcal{B}_{q}$, let%
\begin{equation*}
\eta _{b}(w)\equiv \left[ 1+\left( 2b(w)-1\right) h\right] /2\text{ for }%
w\in \mathcal{W}.
\end{equation*}%
Let $F_{b}$ denote the joint distribution of $\left( Y,W\right) $ such that
under $F_{b}$, the distribution of $W$ satisfies condition (\ref{condition
on the distribution of W}), and $Y$ conditional on $W=w$ follows a Bernoulli
distribution with parameter $\eta _{b}(w)$ for every $w\in \mathcal{W}$. By
the same arguments as those in the proof of Theorem 6 of Massart and N\'{e}d%
\'{e}lec\thinspace (2006, p.\thinspace 2355), we can deduce that, for any
finite subset $\mathcal{C}$ of $\mathcal{B}_{q}$,%
\begin{equation}
\left\{ F_{b}:b\in \mathcal{C}\right\} \subset \mathcal{P}(h,\mathcal{B}_{q})
\label{subset of P(h,Bq)}
\end{equation}%
and 
\begin{equation}
R_{n}(h,\mathcal{B}_{q})\geq \frac{h}{2}\inf_{\widehat{b}\in \mathcal{C}%
}\sup_{b\in \mathcal{C}}E_{F_{b}}\left[ \left\Vert b-\widehat{b}\right\Vert
_{1}\right] .  \label{risk lower bound 1}
\end{equation}

Consider the set%
\begin{equation*}
\mathcal{H}\equiv \left\{ \gamma \in \left\{ -1,0,1\right\} ^{p}:\left\Vert
\gamma \right\Vert _{0}=q\right\} .
\end{equation*}%
By Lemma 4 of \citet{Raskutti2011}, we have that, for $p,q$ even and $q<2p/3$%
, there is a subset $\mathcal{A\subset H}$ with cardinality 
\begin{equation}
\left\vert \mathcal{A}\right\vert \geq \left( \frac{p-q}{q/2}\right) ^{q/2}
\label{cardinality of A}
\end{equation}%
such that%
\begin{equation}
\left\Vert \gamma -\gamma ^{\prime }\right\Vert _{0}\geq q/2\text{ for all }%
\gamma ,\gamma ^{\prime }\in \mathcal{A}\text{ and }\gamma \neq \gamma
^{\prime }.  \label{e1}
\end{equation}

Let%
\begin{equation}
\mathcal{D\equiv }\{1\}\times \{\underline{0}\}\times \epsilon _{n}\mathcal{A%
}  \label{set D}
\end{equation}%
where $\underline{0}=(0,...,0)$ denotes the $k$-dimensional vector of which
all elements take value $0$, and $\epsilon _{n}>0$ is a given sequence that
will be chosen later. Note that, for any $\theta ,\widetilde{\theta }\in 
\mathcal{D}$, 
\begin{equation}
\epsilon _{n}^{2}\left\Vert \theta -\widetilde{\theta }\right\Vert _{0}\leq
\left\Vert \theta -\widetilde{\theta }\right\Vert _{E}^{2}\leq 4\epsilon
_{n}^{2}\left\Vert \theta -\widetilde{\theta }\right\Vert _{0}.
\label{bound for theta minus theta_tilda}
\end{equation}

Now take 
\begin{equation}
\mathcal{C=}\left\{ b_{\theta }\in \mathcal{B}_{q}:\theta \in \mathcal{D}%
\right\} .  \label{set C}
\end{equation}%
We then have that%
\begin{align}
R_{n}(h,\mathcal{B}_{q})& \geq \frac{h}{2}\inf_{\widehat{\theta }\in 
\mathcal{D}}\sup_{\theta \in \mathcal{D}}E_{F_{b_{\theta }}}\left[
\left\Vert b_{\theta }-b_{\widehat{\theta }}\right\Vert _{1}\right]
\label{e} \\
& \geq \frac{hc_{l}}{2}\inf_{\widehat{\theta }\in \mathcal{D}}\sup_{\theta
\in \mathcal{D}}E_{F_{b_{\theta }}}\left[ \left\Vert \theta -\widehat{\theta 
}\right\Vert _{E}\right]  \label{e2} \\
& \geq \frac{hc_{l}\epsilon _{n}}{2}\inf_{\widehat{\theta }\in \mathcal{D}%
}\sup_{\theta \in \mathcal{D}}E_{F_{b_{\theta }}}\left[ \sqrt{\left\Vert
\theta -\widehat{\theta }\right\Vert _{0}}\right]  \label{e3} \\
& \geq \frac{hc_{l}\epsilon _{n}\sqrt{q}}{2\sqrt{2}}\inf_{\widehat{\theta }%
\in \mathcal{D}}\sup_{\theta \in \mathcal{D}}P_{F_{b_{\theta }}}\left( 
\widehat{\theta }\neq \theta \right)  \label{e4} \\
& \geq \frac{hc_{l}\epsilon _{n}\sqrt{q}}{2\sqrt{2}}\inf_{\widehat{\theta }%
\in \mathcal{D}}\left[ 1-\inf_{\theta \in \mathcal{D}}P_{F_{b_{\theta
}}}\left( \widehat{\theta }=\theta \right) \right] ,  \label{e5}
\end{align}%
where the infimum in (\ref{e}) is taken over the set of all estimators $%
\widehat{\theta }$ taking values in $\mathcal{D}$; (\ref{e2}) follows from (%
\ref{subset of P(h,Bq)}), (\ref{set C}) and (\ref{condition on the
distribution of W}); (\ref{e3}) follows from (\ref{bound for theta minus
theta_tilda}); (\ref{e4}) follows from (\ref{e1}).

By Lemma 8 of \citet{massart2006}, we have that, for a given point $%
\widetilde{\theta }\in \mathcal{D}$, 
\begin{equation}
\inf_{\theta \in \mathcal{D}}P_{F_{b_{\theta }}}\left( \widehat{\theta }%
=\theta \right) \leq \phi \vee \frac{\overline{\mathcal{K}}}{\ln \left(
\left\vert \mathcal{D}\right\vert \right) },  \label{e6}
\end{equation}%
where%
\begin{equation*}
\overline{\mathcal{K}}=\frac{n}{\left\vert \mathcal{D}\right\vert -1}%
\sum_{\theta \in \mathcal{D},\theta \neq \widetilde{\theta }}\mathcal{K}%
\left( F_{b_{\theta }},F_{b_{\widetilde{\theta }}}\right) ,
\end{equation*}%
and $\mathcal{K}\left( F_{b_{\theta }},F_{b_{\widetilde{\theta }}}\right) $
is the Kullback-Leibler information between $F_{b_{\theta }}$ and $F_{b_{%
\widetilde{\theta }}}$. For $\widetilde{\theta },\theta \in \mathcal{D}$ and 
$\theta \neq \widetilde{\theta }$, using Lemma 7 of \citet{massart2006}, we
have that, for $h<1$, 
\begin{align*}
\mathcal{K}\left( F_{b_{\theta }},F_{b_{\widetilde{\theta }}}\right) & =h\ln
\left( \frac{1+h}{1-h}\right) \left\Vert b_{\theta }-b_{\widetilde{\theta }%
}\right\Vert _{1} \\
& \leq \frac{2c_{u}h^{2}}{1-h}\left\Vert \theta -\widetilde{\theta }%
\right\Vert _{E} \\
& \leq \frac{4c_{u}h^{2}\epsilon _{n}}{1-h}\sqrt{\left\Vert \theta -%
\widetilde{\theta }\right\Vert _{0}} \\
& \leq \frac{4c_{u}h^{2}\epsilon _{n}\sqrt{2q}}{1-h},
\end{align*}%
where the last inequality follows since $\left\Vert \theta -\widetilde{%
\theta }\right\Vert _{0}\leq 2q$ for all $\widetilde{\theta },\theta \in 
\mathcal{D}$. Hence, we have that%
\begin{equation}
\overline{\mathcal{K}}\leq \frac{4nc_{u}h^{2}\epsilon _{n}\sqrt{2q}}{1-h}.
\label{K_upper_bar}
\end{equation}

Putting together (\ref{e5}), (\ref{e6}) and (\ref{K_upper_bar}), we have
that 
\begin{equation}
R_{n}(h,\mathcal{B}_{q})\geq \frac{h\epsilon _{n}c_{l}\sqrt{q}}{2\sqrt{2}}%
\left( 1-\phi \right)  \label{e7}
\end{equation}%
provided that%
\begin{equation}
\frac{4nc_{u}h^{2}\epsilon _{n}\sqrt{2q}}{\left( 1-h\right) \ln \left(
\left\vert \mathcal{D}\right\vert \right) }\leq \phi .  \label{e8}
\end{equation}%
By (\ref{cardinality of A}) and (\ref{set D}), condition (\ref{e8}) holds
whenever%
\begin{equation}
\frac{8nc_{u}h^{2}\epsilon _{n}\sqrt{2q}}{\left( 1-h\right) q\ln \left( 
\frac{p-q}{q/2}\right) }\leq \phi .  \label{e9}
\end{equation}%
By (\ref{parameter bound}), (\ref{set D}) and (\ref{set C}), we have that $%
\epsilon _{n}\leq \kappa .$ Therefore, we can get the result (\ref{risk
lower bound result}) by setting%
\begin{equation*}
\epsilon _{n}=\frac{\phi \sqrt{q}\left( 1-h\right) \ln \left( \frac{p-q}{q/2}%
\right) }{8nc_{u}h^{2}\sqrt{2}}
\end{equation*}%
provided that this choice of $\epsilon _{n}$ also satisfies that $\epsilon
_{n}\leq \kappa $, which can be easily seen to hold under the condition (\ref%
{lower bound on h}) for the lower bound on the value of $h$.
\end{proof}

\section{The Branch-and-Bound Method for Solving MIO Problems\label%
{bb-method}}

For completeness of the paper and for readers who are unfamiliar with MIO,
we will present briefly the branch-and-bound method for solving the MIO
problem. For further details, see, e.g., \citet{CCZ:2014} for a recent and
comprehensive study on the MIO theory and solution methods.

We take the formulation (\ref{constrained MIO}) as an expositional example
and explain how the branch-and-bound method can be used to solve this MIO
problem. The maximization problem (\ref{constrained MIO}) consists of $n+p$
binary control variables. Let $v=(d_{1},...,d_{n},e_{1},...,e_{p})$ denote
the vector that collects all these binary controls. Let $\widetilde{S}%
_{n}\left( \beta ,\gamma ,v\right) $ denote the objective function of (\ref%
{constrained MIO}). We may maximize $\widetilde{S}_{n}$ over $v$ by
enumerating all possible values of $v$, which amounts to exhaustively
searching over a binary tree that has $2^{n+p}$ leaf nodes. This naive
method is inefficient and becomes practically infeasible for large scale
problems. The branch-and-bound method improves the search efficiency by
avoid visiting those tree nodes which can be fathomed not to constitute the
optimum.

Let $\Gamma _{0}$ denote the space of the controls $\left( \beta ,\gamma
,v\right) $ defined by all the constraints stated in the MIO problem (\ref%
{constrained MIO}). Let $\overline{\Gamma }_{0}$ be an enlargement of $%
\Gamma _{0}$, which is defined analogously to $\Gamma _{0}$ with the
dichotomization constraints $v\in \{0,1\}^{n+p}$ being replaced by the
constraints $v\in \lbrack 0,1]^{n+p}$. Optimizing the objective function $%
\widetilde{S}_{n}$ over $\left( \beta ,\gamma ,v\right) \in \overline{\Gamma 
}_{0}$ reduces to a simple linear programming (LP) problem. Clearly, the
maximized objective value of this LP relaxation problem forms an upper bound
on the function $\widetilde{S}_{n}\left( \beta ,\gamma ,v\right) $ defined
on the original domain $\Gamma _{0}$. Moreover, if the solution for $v$ in
the LP relaxation problem turns out to be a vector of binary values, we can
deduce that the LP relaxation solution for $\left( \beta ,\gamma ,v\right) $
is also the solution to the MIO problem (\ref{constrained MIO}).

When the LP relaxation solution for $v$ contains fractional-valued elements,
we choose a fractional-valued element $v_{j}$ and then construct the two LP
sub-problems, denoted as $LP_{1}$ and $LP_{1}^{\prime }$, which correspond
to maximizing $\widetilde{S}_{n}\left( \beta ,\gamma ,v\right) $ over the
subspaces $\overline{\Gamma }_{1}\equiv $ $\overline{\Gamma }_{0}\cap
\left\{ \left( \beta ,\gamma ,v\right) :v_{j}=0\right\} $ and $\overline{%
\Gamma }_{1}^{\prime }\equiv $ $\overline{\Gamma }_{0}\cap \left\{ \left(
\beta ,\gamma ,v\right) :v_{j}=1\right\} $, respectively. Consider the
problem $LP_{1}$ and note that the treatment of $LP_{1}^{\prime }$ is
similar. There are four possible cases for $LP_{1}$: (i) $\overline{\Gamma }%
_{1}$ is empty and hence $LP_{1}$ is infeasible. (ii) $\overline{\Gamma }%
_{1} $ is non-empty and the maximized objective value of $LP_{1}$ is not
larger than the best known lower bound on the objective value of (\ref%
{constrained MIO}). (iii) $\overline{\Gamma }_{1}$ is non-empty, the
maximized objective value of $LP_{1}$ is larger than the best known lower
bound on the objective value of (\ref{constrained MIO}), and the solution
for $v$ of $LP_{1}$ is in $\{0,1\}^{n+p}$. (iv) $\overline{\Gamma }_{1}$ is
non-empty, the maximized objective value of $LP_{1}$ is larger than the best
known lower bound on the objective value of (\ref{constrained MIO}), and the
solution for $v$ of $LP_{1}$ contains fractional-valued elements.

For cases (i) and (ii), we can bypass further sub-problems of $LP_{1}$ since
these will not yield a solution to the MIO problem (\ref{constrained MIO}).
In other words, all nodes of the binary search tree along the branch implied
by $LP_{1}$ can be pruned and need not be further considered. For case
(iii), we can update the best known feasible solution to the MIO problem (%
\ref{constrained MIO}) as the optimal solution to the problem $LP_{1}$. For
case (iv), the sub-domain $\Gamma _{0}\cap \left\{ \left( \beta ,\gamma
,v\right) :v_{j}=0\right\} $ may still contain an optimal solution.
Therefore, in case (iv), we branch on a fraction-valued component of the $%
LP_{1}$ solution for $v$ to create further two sub-problems and then repeat
this process as described above.

\section{Simulation Results for the DGP(ii) Design}

\label{sim:hetero:design}

In this part of the appendix, we report the simulation results under the
setup of DGP(ii) of Section \ref{simulation study}. Table \ref{tab4} gives
the MIO computation time statistics for solving the PRESCIENCE($q$) problem
under DGP(ii). Compared to the results of Table \ref{tab1}, the PRESCIENCE
problem appeared to be more computationally difficult for the high
dimensional setup in the DGP(ii) design where the maximum computation time
could exceed 2.5 hours. However, the mean and median computation time
remained well capped below 6 minutes across all cases in Table \ref{tab4}.
In fact, the case of the MIO computation lasting over one hour appeared in
only 3 out of the 100 repetitions for the PRESCIENCE(3) simulations in the
setup of $p=60$. 
\begin{table}[tbph]
\caption{PRESCIENCE computation time (CPU seconds) under DGP(ii)}
\label{tab4}
\begin{center}
\begin{tabular}{c|ccc|ccc}
\hline\hline
& \multicolumn{3}{|c|}{$p=10$} & \multicolumn{3}{|c}{$p=60$} \\ \hline
$q$ & 1 & 2 & 3 & 1 & 2 & 3 \\ \hline
mean & 0.39 & 1.55 & 0.73 & 3.47 & 68.53 & \multicolumn{1}{l}{350.4} \\ 
min & 0.05 & 0.04 & 0.03 & 0.21 & 0.03 & \multicolumn{1}{l}{0.04} \\ 
median & 0.40 & 0.76 & 0.38 & 2.77 & 23.07 & \multicolumn{1}{l}{50.13} \\ 
max & 1.02 & 20.34 & 9.97 & 16.56 & 417.2 & \multicolumn{1}{l}{8552} \\ 
\hline
\end{tabular}%
\end{center}
\end{table}

\begin{table}[tbph]
\caption{Comparison of prediction methods under DGP(ii) with $p=10$}
\label{tab5}
\begin{center}
\begin{tabular}{c|ccc|c|ccc|ccc}
\hline\hline
{\small method} & \multicolumn{3}{|c|}{{\small PRESCIENCE(}$q${\small )}} & 
{\small PRE\_CV} & \multicolumn{3}{|c}{\small logit\_lasso} & 
\multicolumn{3}{|c}{\small probit\_lasso} \\ 
& ${\small q=1}$ & ${\small q=2}$ & ${\small q=3}$ &  & ${\small \lambda }%
_{\min }$ & ${\small \lambda }_{1se}$ & ${\small \lambda }_{2se}$ & ${\small %
\lambda }_{\min }$ & ${\small \lambda }_{1se}$ & ${\small \lambda }_{2se}$
\\ \hline
${\small Corr\_sel}$ & {\small 0.86} & {\small 0.95} & {\small 0.95} & 
{\small 0.91} & {\small 0.91} & {\small 0.6} & {\small 0.27} & {\small 0.88}
& {\small 0.55} & {\small 0.28} \\ 
${\small Orac\_sel}$ & {\small 0.86} & {\small 0.01} & {\small 0} & {\small %
0.53} & {\small 0.09} & {\small 0.3} & {\small 0.16} & {\small 0.09} & 
{\small 0.3} & {\small 0.16} \\ 
${\small Num\_irrel}$ & {\small 0.14} & {\small 1.04} & {\small 2.02} & 
{\small 0.67} & {\small 2.93} & {\small 0.62} & {\small 0.26} & {\small 2.49}
& {\small 0.51} & {\small 0.26} \\ 
${\small in\_Score}$ & {\small 0.834} & {\small 0.871} & {\small 0.894} & 
{\small 0.860} & {\small 0.787} & {\small 0.710} & {\small 0.630} & {\small %
0.776} & {\small 0.695} & {\small 0.627} \\ 
${\small in\_RS}$ & {\small 1.095} & {\small 1.144} & {\small 1.175} & 
{\small 1.131} & {\small 1.031} & {\small 0.930} & {\small 0.825} & {\small %
1.016} & {\small 0.910} & {\small 0.822} \\ 
${\small out\_Score}$ & {\small 0.724} & {\small 0.711} & {\small 0.696} & 
{\small 0.716} & {\small 0.673} & {\small 0.614} & {\small 0.553} & {\small %
0.668} & {\small 0.606} & {\small 0.552} \\ 
${\small out\_RS}$ & {\small 0.948} & {\small 0.930} & {\small 0.910} & 
{\small 0.937} & {\small 0.881} & {\small 0.804} & {\small 0.724} & {\small %
0.874} & {\small 0.794} & {\small 0.723} \\ \hline
\end{tabular}%
\end{center}
\end{table}

\begin{table}[tbph]
\caption{Comparison of prediction methods under DGP(ii) with $p=60$}
\label{tab6}
\begin{center}
\begin{tabular}{c|ccc|c|ccc|ccc}
\hline\hline
{\small method} & \multicolumn{3}{|c|}{{\small PRESCIENCE(}$q${\small )}} & 
{\small PRE\_CV} & \multicolumn{3}{|c}{\small logit\_lasso} & 
\multicolumn{3}{|c}{\small probit\_lasso} \\ 
& ${\small q=1}$ & ${\small q=2}$ & ${\small q=3}$ &  & ${\small \lambda }%
_{\min }$ & ${\small \lambda }_{1se}$ & ${\small \lambda }_{2se}$ & ${\small %
\lambda }_{\min }$ & ${\small \lambda }_{1se}$ & ${\small \lambda }_{2se}$
\\ \hline
${\small Corr\_sel}$ & {\small 0.76} & {\small 0.82} & {\small 0.88} & 
{\small 0.82} & {\small 0.64} & {\small 0.45} & {\small 0.29} & {\small 0.63}
& {\small 0.43} & {\small 0.28} \\ 
${\small Orac\_sel}$ & {\small 0.76} & {\small 0} & {\small 0} & {\small 0.41%
} & {\small 0} & {\small 0.13} & {\small 0.16} & {\small 0} & {\small 0.11}
& {\small 0.15} \\ 
${\small Num\_irrel}$ & {\small 0.24} & {\small 1.18} & {\small 2.12} & 
{\small 0.88} & {\small 4.53} & {\small 1.17} & {\small 0.47} & {\small 4.17}
& {\small 1.06} & {\small 0.44} \\ 
${\small in\_Score}$ & {\small 0.842} & {\small 0.894} & {\small 0.927} & 
{\small 0.880} & {\small 0.766} & {\small 0.680} & {\small 0.629} & {\small %
0.759} & {\small 0.673} & {\small 0.626} \\ 
${\small in\_RS}$ & {\small 1.103} & {\small 1.171} & {\small 1.216} & 
{\small 1.153} & {\small 1.000} & {\small 0.887} & {\small 0.821} & {\small %
0.990} & {\small 0.878} & {\small 0.817} \\ 
${\small out\_Score}$ & {\small 0.713} & {\small 0.693} & {\small 0.673} & 
{\small 0.700} & {\small 0.608} & {\small 0.575} & {\small 0.540} & {\small %
0.608} & {\small 0.571} & {\small 0.539} \\ 
${\small out\_RS}$ & {\small 0.934} & {\small 0.907} & {\small 0.881} & 
{\small 0.917} & {\small 0.797} & {\small 0.753} & {\small 0.708} & {\small %
0.796} & {\small 0.748} & {\small 0.706} \\ \hline
\end{tabular}%
\end{center}
\end{table}

We compare in Tables \ref{tab5} and \ref{tab6} the predictive and variable
selection performance results for the various prediction methods given in (%
\ref{prediction methods}). For the penalized MLE approaches, the
logit\_lasso and probit\_lasso implemented with $\lambda _{\min }$ performed
better in terms of predictive performance than those implemented with $%
\lambda _{1se}$ or $\lambda _{2se}$. Yet, the PRE\_CV approach still had the
best overall performance among all the prediction approaches in Tables \ref%
{tab5} and \ref{tab6}. We also note that the PRE\_CV approach could
outperform the logit\_lasso and probit\_lasso approaches by a large margin
in both the in-sample and out-of-sample predictive performances in the high
dimensional variable selection setup.

It is well known in the binary prediction literature 
\citep[see
e.g.][]{Elliott2013} that the optimal prediction rule in terms of score
maximization does not hinge on knowing the true distribution of $Y$ given $W$%
, and binary prediction based on the MLE approach with a misspecified
likelihood can yield poor predictive performance. For the DGP(ii) design,
the binary response probability $P(Y=1|W)$ depends on the index $\left[
\sigma (W)\right] ^{-1}W^{\prime }\theta $, which is nonlinear in the
variables $V_{1}$ and $V_{2}$, such that the logit and probit likelihoods
with an index linear in $W$ are misspecified. We approximated this
nonlinearity by using covariates that consisted of cubic polynomial terms in 
$\left( V_{1},V_{2}\right) $. Specifically, we replaced the last 7 variables
of the original auxiliary covariates $\left( V_{2},...,V_{p+1}\right) $ by
the vector 
\begin{equation*}
(V_{1}^{2},V_{2}^{2},V_{1}^{3},V_{2}^{3},V_{1}V_{2},V_{1}V_{2}^{2},V_{1}^{2}V_{2}),
\end{equation*}%
where each nonlinear covariate was standardized to have mean zero and
variance unity, so that the resulting auxiliary covariate vector remained to
be of dimension $p$. In Table \ref{tab_nonlinear_covariates}, we reported
simulation results under DGP(ii) with this auxiliary covariate specification
on the predictive performance comparison between the PRESCIENCE and
penalized MLE approaches. Except for the auxiliary covariate setting, all
simulation setups for the results of Table \ref{tab_nonlinear_covariates}
were the same as described in Section \ref{simulation study}. Comparing
these results to those of Tables \ref{tab5} and \ref{tab6}, we found that
the out-of-sample predictive performance for both the logit\_lasso and
probit\_lasso approaches could indeed be improved; however, the PRE\_CV
approach continued dominating these MLE approaches and the performance gain
remained substantial in the setup with 60 auxiliary covariates.

\begin{table}[tbph]
\caption{Predictive performance comparison under DGP(ii) where the auxiliary
covariates include cubic polynomial terms in $(V_{1},V_{2})$. }
\label{tab_nonlinear_covariates}
\begin{center}
\begin{tabular}{c|ccc|c|ccc|ccc}
\hline\hline
{\small method} & \multicolumn{3}{|c|}{{\small PRESCIENCE(}$q${\small )}} & 
{\small PRE\_CV} & \multicolumn{3}{|c}{\small logit\_lasso} & 
\multicolumn{3}{|c}{\small probit\_lasso} \\ 
& ${\small q=1}$ & ${\small q=2}$ & ${\small q=3}$ &  & ${\small \lambda }%
_{\min }$ & ${\small \lambda }_{1se}$ & ${\small \lambda }_{2se}$ & ${\small %
\lambda }_{\min }$ & ${\small \lambda }_{1se}$ & ${\small \lambda }_{2se}$
\\ \hline
\multicolumn{11}{c}{\small 10 auxiliary covariates} \\ \hline
${\small in\_Score}$ & {\small 0.837} & {\small 0.879} & {\small 0.911} & 
{\small 0.876} & {\small 0.801} & {\small 0.721} & {\small 0.661} & {\small %
0.791} & {\small 0.723} & {\small 0.664} \\ 
${\small in\_RS}$ & {\small 1.099} & {\small 1.155} & {\small 1.198} & 
{\small 1.152} & {\small 1.051} & {\small 0.947} & {\small 0.867} & {\small %
1.039} & {\small 0.949} & {\small 0.870} \\ 
${\small out\_Score}$ & {\small 0.717} & {\small 0.720} & {\small 0.707} & 
{\small 0.718} & {\small 0.701} & {\small 0.634} & {\small 0.578} & {\small %
0.699} & {\small 0.636} & {\small 0.582} \\ 
${\small out\_RS}$ & {\small 0.939} & {\small 0.942} & {\small 0.926} & 
{\small 0.940} & {\small 0.918} & {\small 0.830} & {\small 0.757} & {\small %
0.914} & {\small 0.832} & {\small 0.762} \\ \hline
\multicolumn{11}{c}{\small 60 auxiliary covariates} \\ \hline
${\small in\_Score}$ & {\small 0.846} & {\small 0.900} & {\small 0.936} & 
{\small 0.89} & {\small 0.779} & {\small 0.678} & {\small 0.629} & {\small %
0.769} & {\small 0.674} & {\small 0.625} \\ 
${\small in\_RS}$ & {\small 1.107} & {\small 1.180} & {\small 1.227} & 
{\small 1.166} & {\small 1.018} & {\small 0.883} & {\small 0.820} & {\small %
1.004} & {\small 0.879} & {\small 0.815} \\ 
${\small out\_Score}$ & {\small 0.696} & {\small 0.699} & {\small 0.695} & 
{\small 0.696} & {\small 0.616} & {\small 0.576} & {\small 0.540} & {\small %
0.615} & {\small 0.574} & {\small 0.538} \\ 
${\small out\_RS}$ & {\small 0.912} & {\small 0.915} & {\small 0.911} & 
{\small 0.911} & {\small 0.807} & {\small 0.755} & {\small 0.708} & {\small %
0.805} & {\small 0.752} & {\small 0.705} \\ \hline
\end{tabular}%
\end{center}
\end{table}

\section{Additional Simulations on the Performance of the Warm-Start MIO
Approaches to the PRESCIENCE Problem\label{simulations of the cold and warm
start methods}}

In Appendix \ref{simulations of the cold and warm start methods}, we conduct
a simulation study on the performance of adopting the warm-start strategy of
Section \ref{sec:para-space} in solving the MIO formulations (\ref%
{constrained MIO}) and (\ref{Florios-Skouras MIO}). We used the setup of
DGP(ii) of Section \ref{simulation study} for the simulation design. For all
simulation experiments in this section, we used a training sample of $n=100$
observations over which we computed the exact solutions to all the MIO
problems. We set $p$, the dimension of the vector of auxiliary covariates,
to be 10. We used the space (\ref{PRESCIENCE parameter space}) as the
parameter space $\Theta $ for the cold-start MIO solution approach. For the
warm-start MIO formulations, we set $\tau =1.5$ and constructed $(\widehat{P}%
_{i})_{i=1}^{n}$ using the fitted choice probabilities from the logit
regression of $Y$ on the entire covariate vector $W$ to derive the space $%
\widehat{\Theta }\left( \tau \right) $ as a refinement of the initial
parameter space $\Theta $. The number of simulation repetition was set to be 
$100$.

We now present the simulation results. Table \ref{tab7} gives the summary
statistics of the MIO computation time in CPU seconds and the average of the
maximized scores over all the simulation repetitions. From this table, we
note that, fixing the start method, the average of maximized scores under
the MIO formulation (\ref{constrained MIO}) was identical to that under the
MIO formulation (\ref{Florios-Skouras MIO}). In fact, the maximized score
values computed under the two MIO formulations were also identical across
all the simulation repetitions. This matched the mathematical equivalence
between the MIO problems (\ref{constrained MIO}) and (\ref{Florios-Skouras
MIO}). Across the start method, we find that the warm start approach, which
is based on a smaller parameter space, could miss the global optimum.
However, the difference of the cold and warm-start based maximized objective
values was very small. In fact, in no more than 7\% of the simulation
repetitions did we observe the occurrence of such differences among which
the maximal difference was about 0.02.

On the other hand, we observe significant reduction of computation time from
employing the warm start strategy across nearly all cases in Table \ref{tab7}%
. However, the cold start approach might sometimes be less computationally
costly than its warm start version. Yet this happened only in the simulation
with the MIO formulation (\ref{constrained MIO}) and its ocurrence was rare
(in only 1 out of 100 repetitions). Therefore, we believe that the warm
start method may be a useful heuristic device to improve the computational
efficiency for the MIO based computation of PRESCIENCE. Finally, concerning
the performance comparison of the MIO formulations (\ref{constrained MIO})
and (\ref{Florios-Skouras MIO}), Table \ref{tab7} reveals that the
formulation (\ref{constrained MIO}) tended to outperform the formulation (%
\ref{Florios-Skouras MIO}) in terms of computation time in the cold-start
setting. This tendency appeared to be reversed in the warm-start setting. 
\begin{table}[tbph]
\caption{Performance comparison of cold and warm-start MIO formulations}
\label{tab7}
\begin{center}
\begin{tabular}{c|ccc|ccc}
\hline\hline
& \multicolumn{3}{c|}{cold start} & \multicolumn{3}{c}{warm start} \\ \hline
$q$ & 1 & 2 & 3 & 1 & 2 & 3 \\ \hline
\multicolumn{7}{l}{MIO formulation (\ref{constrained MIO})} \\ \hline
maximized score & 0.812 & 0.837 & 0.855 & 0.811 & 0.836 & 0.854 \\ \hline
\multicolumn{7}{c}{MIO computation time} \\ \hline
mean & 1.44 & 40.5 & 149 & 0.49 & 4.30 & 43.4 \\ 
min & 0.64 & 1.68 & 0.57 & 0.07 & 0.11 & 0.18 \\ 
median & 1.33 & 33.5 & 94.5 & 0.36 & 2.11 & 6.54 \\ 
max & 3.18 & 223 & 1102 & 1.31 & 59.7 & 3074 \\ \hline
\multicolumn{7}{l}{MIO formulation (\ref{Florios-Skouras MIO})} \\ \hline
maximized score & 0.812 & 0.837 & 0.855 & 0.811 & 0.836 & 0.854 \\ \hline
\multicolumn{7}{c}{MIO computation time} \\ \hline
mean & 7.78 & 79.2 & 401 & 0.55 & 4.32 & 19.4 \\ 
min & 0.55 & 1.09 & 0.42 & 0.05 & 0.09 & 0.11 \\ 
median & 4.66 & 63.2 & 242 & 0.40 & 1.92 & 5.05 \\ 
max & 31.3 & 868 & 3153 & 4.28 & 35.7 & 367 \\ \hline
\end{tabular}%
\end{center}
\end{table}

\section{Empirical Illustration for Best Subset Selection Using Linear
Specification\label{emp:linear:spec}}

In this section of the appendix, we report empirical results using the
linear specification of covariates. Specifically, we use the same focused
covariates as constructed in the main text; for auxiliary covariates $Z$, we
consider the simple specification where $Z=\left( CARS,DOVTT,DIVTT\right) $.
Under this setup, we have that $k=1$ and $p=3$. Following %
\citet{Florios:Skouras:08}, we set all unknown parameters to be within the
range $[-10,10]$ and took the parameter space $\Theta $ to be $\left[ -10,10%
\right] ^{4}$. The refined space $\widehat{\Theta }\left( \tau \right) $ was
computed accordingly as described in Section \ref{sec:para-space}. Table \ref%
{tab-01} presents the resulting refined parameter bounds derived from $%
\widehat{\Theta }\left( \tau \right) $. 
\begin{table}[tbh]
\caption{Refined parameter bounds ($\protect\tau =1.5$)}
\label{tab-01}
\begin{center}
\emph{Covariate specification}: $k=1,p=3$ \\[0pt]
\vspace*{2ex} 
\begin{tabular}{lcc}
\hline\hline
Variable & lower bound & upper bound \\ \hline
Intercept & -7.8275 & 7.8275 \\ 
$CARS$ & -5.4143 & 5.4143 \\ 
$DOVTT$ & -1.9229 & 1.9229 \\ 
\ $DIVTT$ & -0.7884 & 0.7884 \\ \hline
\end{tabular}%
\end{center}
\end{table}

We can clearly see from Table \ref{tab-01} that using the refined space $%
\widehat{\Theta }\left( \tau \right) $ helped to reduce the parameter search
space in both of the MIO formulations (\ref{constrained MIO}) and (\ref%
{Florios-Skouras MIO}). The extent of this reduction could be quite large
even when the enlargement parameter $\tau $ was set to be $1.5$. In fact,
the size of $\widehat{\Theta }\left( \tau \right) $ was merely about 0.64\%
of that of $\Theta $. Therefore, we can anticipate considerably
computational efficiency gain from using the warm-start MIO formulations.

In Table \ref{tab-02}, we present comparative results of the warm-start and
cold-start approaches for the MIO formulations in (\ref{constrained MIO})
specified with different values of the cardinality bound $q$. Since there
are only 4 unknown parameters in this simple setup, we set $\varepsilon =0$
in (\ref{approximate MIO solution}) and solved for the exact PRESCIENCE. 
\begin{table}[tbh]
\caption{Implementation Using MIO formulation (\protect\ref{constrained MIO}%
) }
\label{tab-02}
\begin{center}
\emph{Covariate specification}: $k=1,p=3$ \\[0pt]
\vspace*{2ex} 
\begin{tabular}{|l|l|l|l|l|l|l|}
\hline\hline
$q$ & \multicolumn{2}{|l}{$1$} & \multicolumn{2}{|l}{$2$} & 
\multicolumn{2}{|l|}{$3$} \\ \hline
MIO start method & warm & cold & warm & cold & warm & cold \\ \hline
\multicolumn{7}{|l|}{focus covariates} \\ \hline
$DCOST$ & 1 & 1 & 1 & 1 & 1 & 1 \\ \hline
Intercept & 3.4654 & 3.4258 & 3.2447 & 5.2013 & 4.7953 & 4.9267 \\ \hline
\multicolumn{7}{|l|}{auxiliary covariates} \\ \hline
$CARS$ & 2.7420 & 2.7191 & 2.7470 & 4.3828 & 3.6571 & 3.8179 \\ \hline
$DOVTT$ & 0 & 0 & 0 & 0.9278 & 0.7951 & 0.7952 \\ \hline
$DIVTT$ & 0 & 0 & 0.6653 & 0 & 0.3830 & 0.4978 \\ \hline
\multicolumn{7}{|l|}{in-sample performance} \\ \hline
maximized score & 756 & 756 & 763 & 763 & 765 & 765 \\ \hline
maximized average score & 0.8979 & 0.8979 & 0.9062 & 0.9062 & 0.9086 & 0.9086
\\ \hline
\multicolumn{7}{|l|}{MIO solver output} \\ \hline
$MIO\_gap$ & 0 & 0 & 0 & 0 & 0 & 0 \\ \hline
CPU time (in seconds) & 15 & 13 & 112 & 903 & 253 & 1887 \\ \hline
branch-and-bound nodes & 6321 & 6120 & 62630 & 304708 & 129947 & 887349 \\ 
\hline
\end{tabular}%
\end{center}
\end{table}

To interpret the results of Table \ref{tab-02}, first note that $MIO\_gap=0$
for all MIO problems in this table. Thus all these MIO solutions were exact;
moreover, both the cold-start and warm-start MIO approaches yielded the same
maximized objective values and the parameter estimates for auxiliary
covariates indeed respected the $\ell _{0}$-norm constraint specified in (%
\ref{theta_q}). The parameter estimates solved by these two different
approaches were very similar for the cases of $q\in \{1,3\}$. Since the
maximum score objective function is a step function, it is not surprising to
have multiple solutions and thus the PRESCIENCE for a given value of $q$
need not be unique. This can be clearly seen from the case of $q=2$, where
the results of the cold-start and warm-start methods differed in the
covariate to be excluded from the corresponding PRESCIENCE.

We now assess the computational efficiency of the warm-start and cold-start
approaches. From Table \ref{tab-02}, we can see that both approaches
performed very well. Most of the MIO cases considered in this table were
solved in few minutes and the case taking the longest time was also solved
in about half an hour. We also notice from these results that, except for
the case of $q=1$ under which both approaches were comparable, the
cold-start formulation was clearly outperformed by its corresponding
warm-start version. The difference in computational efficiency can be
sizable: the warm-start approach just took about 12\% (13\%) of the time and
20\% (14\%) of the branch-and-bound nodes used by the cold-start approach to
solve the $q=2$ ($q=3$) case.

In this empirical application, the computational merit of using the refined
parameter space is also evident for the PRESCIENCE implementation using the
MIO formulation (\ref{Florios-Skouras MIO}). The results for this
formulation are summarized in Table \ref{tab-03}. 
\begin{table}[tbh]
\caption{Implementation Using MIO formulation (\protect\ref{Florios-Skouras
MIO})}
\label{tab-03}
\begin{center}
\emph{Covariate specification}: $k=1,p=3$ \\[0pt]
\vspace*{2ex} 
\begin{tabular}{|l|l|l|l|l|l|l|}
\hline\hline
$q$ & \multicolumn{2}{|l}{$1$} & \multicolumn{2}{|l}{$2$} & 
\multicolumn{2}{|l|}{$3$} \\ \hline
MIO start method & warm & cold & warm & cold & warm & cold \\ \hline
\multicolumn{7}{|l|}{focus covariates} \\ \hline
$DCOST$ & 1 & 1 & 1 & 1 & 1 & 1 \\ \hline
Intercept & 3.2803 & 3.2803 & 4.6701 & 5.2013 & 4.8270 & 4.9267 \\ \hline
\multicolumn{7}{|l|}{auxiliary covariates} \\ \hline
$CARS$ & 2.4666 & 2.4667 & 3.4212 & 4.3828 & 3.7000 & 3.8179 \\ \hline
$DOVTT$ & 0 & 0 & 0.9542 & 0.9278 & 0.8141 & 0.7952 \\ \hline
$DIVTT$ & 0 & 0 & 0 & 0 & 0.3211 & 0.4978 \\ \hline
\multicolumn{7}{|l|}{in-sample performance} \\ \hline
maximized score & 756 & 756 & 763 & 763 & 765 & 765 \\ \hline
maximized average score & 0.8979 & 0.8979 & 0.9062 & 0.9062 & 0.9086 & 0.9086
\\ \hline
\multicolumn{7}{|l|}{MIO solver output} \\ \hline
$MIO\_gap$ & 0 & 0 & 0 & 0 & 0 & 0 \\ \hline
CPU time (in seconds) & 62 & 470 & 81 & 6428 & 144 & 2052 \\ \hline
branch-and-bound nodes & 38163 & 139833 & 89977 & 4867501 & 154390 & 1539205
\\ \hline
\end{tabular}%
\end{center}
\end{table}

As in Table \ref{tab-02}, the parameter estimates from both the warm-start
and cold-start approaches in Table \ref{tab-03} were very similar for the
cases $q\in \left\{ 1,3\right\} $. For these two cases, the parameter
estimates as displayed in Tables \ref{tab-02} and \ref{tab-03} were also
qualitatively and quantitatively similar across the MIO formulations (\ref%
{constrained MIO}) and (\ref{Florios-Skouras MIO}). For the case of $q=2$,
in contrast to the results for the formulation (\ref{constrained MIO}), the
variable being excluded was the same for both warm-start and cold-start
versions of the formulation (\ref{Florios-Skouras MIO}). We notice that the
case of $q=3$ reduces to the maximum score estimation problem using all
covariates; for this case, our MIO estimates were quite similar to those
computed by \citet{Florios:Skouras:08} and both our and their maximized
objective values were identical.

It can be noticed that, across all cases in Tables \ref{tab-02} and \ref%
{tab-03}, the variable $CARS$ was always selected and its parameter estimate
was of a much larger magnitude than those of other selected auxiliary
covariates. This indicates that $CARS$ is the most important variable among
the three auxiliary covariates. Moreover, by comparing the maximized average
scores derived under difference cases of $q$, there was very little loss in
the goodness of fit from adopting the parsimonious specification using only $%
CARS$ as the auxiliary covariate.

Regarding the computational efficiency, we can clearly see from the results
across all cases of $q$ in Table \ref{tab-03} that, for the MIO formulation (%
\ref{Florios-Skouras MIO}), there was huge performance gain from using the
warm-start approach in terms of reduction of the CPU time and
branch-and-bound nodes. Putting together the results from both Tables 2 and
3, we thus find that, regardless of the MIO formulations (\ref{constrained
MIO}) and (\ref{Florios-Skouras MIO}), it is generally far more
computationally efficient to adopt the warm-start strategy in the
implementation.

We now compare computational performance across the formulations (\ref%
{constrained MIO}) and (\ref{Florios-Skouras MIO}). From Tables \ref{tab-02}
and \ref{tab-03}, we can see that, for all three cases of $q$, it took far
fewer branch-and-bound nodes to solve the formulation (\ref{constrained MIO}%
) than to solve the formulation (\ref{Florios-Skouras MIO}). However, in
terms of usage of the CPU time, the former formulation did not completely
dominate the latter. This can be intuitively explained as follows.

There are $2n$ inequalities stated in (\ref{constraint on di}) of the
formulation (\ref{constrained MIO}). By contrast, the number of inequalities
in (\ref{sign matching constraints}) of the formulation (\ref%
{Florios-Skouras MIO}) is only half of that amount. Hence, the corresponding
linear programming (LP) relaxation problems in the branch-and-bound solution
procedure for the formulation (\ref{constrained MIO}) are likely to be
tighter than those for the formulation (\ref{Florios-Skouras MIO}). This
would help to reduce the number of branching steps required to reach the
optimum. On the other hand, for the formulation (\ref{constrained MIO}), it
may take much longer to solve at each node the LP relaxation problem which
contains a massive amount of inequality constraints. Thus, there is a
tradeoff between the computational cost per node and the total number of
required nodes in the solution procedure. This tradeoff may depend on the
sample size, the support of the data and the variable selection bound $q$.
Therefore, we find that each of these two MIO formulations has its strength
and hence both complement each other for solving the PRESCIENCE problems.

{\singlespacing
\bibliographystyle{econometrica}
\bibliography{BSBP}
}

\end{document}